\documentclass[msom,nonblindrev]{informs3}
\OneAndAHalfSpacedXI

\usepackage{makecell}
\usepackage{multirow}
\usepackage{natbib}
 \bibpunct[, ]{(}{)}{,}{a}{}{,}%
 \def\bibfont{\small}%

\usepackage[utf8]{inputenc} 
\usepackage[T1]{fontenc}    
\usepackage{hyperref}       
\usepackage{url}            
\usepackage{booktabs}       
\usepackage{amsfonts}       
\usepackage{nicefrac}       
\usepackage{microtype} 
\usepackage{enumerate}
\usepackage{comment} 
\newcommand{\N}{\mathbb{N}}

\newcommand{\I}{\mathbb{I}}
\newcommand{\R}{\mathbb{R}}
\newcommand{\prob}{\mathbb{P}}

\newcommand{\E}{\mathbb{E}}

\newtheorem{fact}{Fact}
\usepackage{bm}

\usepackage{amsmath,amssymb}
\usepackage{array}
\def\defeq{:=}
\usepackage[textsize=tiny]{todonotes}

\newcommand{\sellreg}{\text{Reg}_{\text{sell}}}
\newcommand{\buyreg}{\text{Reg}_{\text{buy}}}

\newcommand{\J}{J}

\newcommand{\rev}{\pi}
\newcommand{\util}{U}
\newcommand{\V}{V}
\newcommand{\D}{D}
\newcommand{\term}{m^{*}}
\newcommand{\med}{\text{med}}
\newcommand{\iter}{\text{iter}}
\newcommand{\lef}{\text{L}}
\newcommand{\rig}{\text{R}}

\newcommand{\opt}{\textsc{OPT}}

\newcommand{\valp}{\bm{g}}
\newcommand{\valpe}{g}
\newcommand{\thr}{\psi}

\newcommand{\norm}[1]{\lVert#1\rVert}

\newcommand{\tote}{L}

\newcommand{\jroi}{r}
\newcommand{\jbud}{b}
\newcommand{\ROI}{\texttt{R}}
\newcommand{\bud}{\texttt{B}}

\newcommand{\budprim}{\texttt{P-Budget}}
\newcommand{\roiprim}{\texttt{P-ROI}}

\usepackage{algorithm}
\usepackage{algpseudocode}
\usepackage{amsmath,amssymb,amsfonts}
\usepackage{enumitem}
\usepackage{bm}
\usepackage{thmtools,thm-restate}
\usepackage{cleveref}
\usepackage{subcaption}
\usepackage{graphicx}

\newcommand{\cond}[1]{\left. \right|}
\def \rev{{\sf{\small{rev}}}}

\newtheorem{theorem}{Theorem}[section]

\newtheorem{lemma}[theorem]{Lemma}
\newtheorem{remark}{Remark}[section]
\newtheorem{proposition}[theorem]{Proposition}
\newtheorem{definition}{Definition}[section]

\newtheorem{assumption}{Assumption}[section]

\begin{document}

\TITLE{Learning to Price against a Budget and ROI Constrained Buyer}
\ARTICLEAUTHORS{%
\AUTHOR{Negin Golrezaei}
\AFF{Sloan School of Management, Massachusetts Institute of Technology, \EMAIL{golrezae@mit.edu} \URL{}}
\AUTHOR{Patrick Jaillet}
\AFF{Department of Electrical Engineering and Computer Science, Massachusetts Institute of Technology, \EMAIL{jaillet@mit.edu} \URL{}}
\AUTHOR{Jason Cheuk Nam Liang}
\AFF{Operations Research Center, Massachusetts Institute of Technology, \EMAIL{jcnliang@mit.edu} \URL{}}
\AUTHOR{Vahab Mirrokni}
\AFF{Google Research, \EMAIL{mirrokni@google.com}\URL{}}
} 

\ABSTRACT{Internet advertisers (buyers) repeatedly procure ad impressions from ad platforms (sellers) with the aim to maximize total conversion (i.e. ad value) while respecting both budget and return-on-investment (ROI) constraints for efficient utilization of limited monetary resources. Facing such a constrained buyer who aims to learn her optimal strategy to acquire impressions, we study from a seller's perspective how to learn and price ad impressions through repeated posted price mechanisms to maximize revenue. For this two-sided learning setup, we propose a learning algorithm for the seller that utilizes an episodic binary-search procedure to identify a revenue-optimal selling price.  We show that such a simple learning algorithm enjoys low seller regret when within each episode, the budget and ROI constrained buyer approximately best responds to the posted price.  We present simple yet natural buyer's bidding algorithms under which the buyer approximately best responds while satisfying budget and ROI constraints, leading to a low regret for our proposed seller pricing algorithm. The design of our seller algorithm is motivated by the fact that the seller's revenue function admits a bell-shaped structure when the buyer best responds to prices under budget and ROI constraints, enabling our seller algorithm to identify revenue-optimal selling prices efficiently.
}

\KEYWORDS{online advertising, pricing, return-on-investment, online learning, mechanism design}
\maketitle

\section{Introduction}
\label{sec:intro}
In online advertising markets, advertisers (i.e. buyers) run ad campaigns by procuring ad impressions in selling mechanisms held by the platform (i.e. seller). To efficiently utilize limited monetary resources that are allocated to a certain campaign, 
advertisers' strategies in the procurement process are typically subject to financial constraints, which generally include budget and \emph{return-on-investment} (ROI) constraints. Budget constraints primarily reflect advertisers' monetary limits due to organizational planning, whereas
ROI constraints enforces  the desired performance/return  on the amount of capital spent \cite{kireyev2016display, golrezaei2018auction, balseiro2019black}. The presence of such financial constraints, along with the increasing availability of real time data, motivates buyers' deployment of complex algorithms to procure impressions. Such financial constraint and algorithm driven buyer behavior introduces significant challenges to sellers' design of selling mechanisms, primarily due to the fact that buyer algorithms adapt quickly and constantly to data generated by buyer-seller interactions, and also sellers' lack of information on buyers' model primitives such as target ROI, budget, buyer algorithm, etc. In this work, we address the following question:

\textit{From the perspective of a seller (e.g. ad platform), what is an optimal selling strategy against a buyer who adopts value-maximizing algorithms under both budget and ROI constraints?}

 We study the setting where a seller repeatedly sells items to a single budget and ROI constrained buyer through a posted price mechanism. This single-buyer setup is primarily motivated by ad platforms' targeting practices that enable advertisers to target users who may be more interested in their ads, as such practices along with advertisers’ heterogeneous targeting criteria lead to a very small number of advertisers/buyers per ad impression, justifying our single-buyer setup. Throughout the repeated mechanism, the seller posts a price for an impression during each period, and the buyer decides on whether to accept and pay the posted price for the sold impression. Our key focus lies in the practical two-sided learning setup where buyers adopt learning algorithms under both budget and ROI constraints, and the seller sets prices algorithmically based on past transactions. The key challenge for the seller's problem of interest is two-fold: the seller does not know the buyer parameters such as target ROI, budget or algorithm, and buyer actions are ``non-stationary'' as buyer algorithms constantly adapt to the past buyer-seller algorithmic interactions. The goal of this work is to design a revenue-maximizing seller pricing strategy against algorthmic and financially constrained buyers in such a limited information setting.

 The main contribution of this work is that we propose a simple seller algorithm that does not require explicitly learning buyer's parameters nor reverse engineering the buyer's learning algorithms. We show that our algorithm is feasible in achieving high revenue under limited information by exploiting a salient property of the seller revenue function against financially-constrained buyers. In particular, we summarize our contributions as followed:


\paragraph{Main contributions.}
We first characterize the seller revenue function against a clairvoyant budget and ROI constrained buyer who always best responds to posted prices. To begin with, we start by characterizing the buyer's best response to a posted price is a ``threshold strategy'' which instructs the buyer to accept the sold item if her valuation exceeds a certain threshold that depends on the posted price. With this characterization of buyer best response, we show that the seller revenue function against a best-responding clairvoyant buyer admits a salient ``bell-shaped'' structure: as the seller increases prices, the corresponding per-period seller revenue first monotonically increases and decreases. We argue that such a structure is exploitable by the seller to extract revenue even without knowing buyer model primitives such as value distribution, budget rate, and target ROI. 

We exploit this bell-shaped structure and
design an episodic binary search seller pricing algorithm. In each episode, the algorithm sets a single price, and then moves on to the next episode with an updated price based on a binary search procedure w.r.t. the realized revenue of previous prices.  
We also characterize general buyer-algorithm adaptiveness properties that allow buyers to adapt quickly to prices in seller episodes, and present regret analyses against buyer algorithms that are adaptive to seller prices in the sense of our defined adaptiveness properties. Moreover, we argue that seller regret of our proposed algorithm is driven by the agent (i.e. 
seller or buyer) who incurs a larger loss in terms of learning error.

Finally, we analyze example buyer algorithm examples that satisfy the aforementioned adaptiveness properties, while simultaneously aim to maximize total value under both budget and ROI constraints. In particular, we consider clairvoyant buyers who best respond in each period, as well as buyers who make decisions based on machine-learned advice that take the form of value distribution estimates. For each of these buyers, we show that both buyer and seller regret are sublinear.
\subsection{Literature review}

\paragraph{Mechanism design for budget and ROI constrained buyers.}
One relevant line of research addresses the mechanism design problem for budget or ROI constrained buyers.
As one of the pioneering works regarding mechanism for financially constrained buyers,  \cite{laffont1996optimal} derives the optimal mechanism for symmetric buyers and public budget information. 
On the contrary, a more recent paper \cite{pai2014optimal} studies the general multidimensional mechanism design setting against buyers with private budgets. Regarding ROI constrained buyers, \cite{golrezaei2018auction} shows that the optimal mechanism for symmetric ROI-constrained buyers is either second-price auctions with reduced reserve prices or subsidized second-price auctions. The work also derives an optimal mechanism for asymmetric ROI buyers. There is also a wide range of work that study dynamic mechanism design for budget constrained buyers, and we refer the reader to the survey \cite{bergemann2010dynamic} and references therein. There have also been recent developments for designing auctions in a setup called \textit{autobidding}, where advertisers simultaneously participate in parallel auction to maximize total value while subject to a coupled ROI constraint across all auctions (see e.g. \cite{aggarwal2019autobidding,deng2021towards,balseiro2021robust,deng2022fairness}). All aforementioned works focus on the static mechanism design problem, whereas in this paper we address the topic of designing repeated posted price mechanisms to sell to both budget and ROI constrained buyers.

\paragraph{Selling to strategic or learning buyers.}
\cite{kleinberg2003value} studies the scenario where the seller sells items through a repeated posted price mechanism to a single truthful buyer who simply accepts the price if her valuation is greater than the offered price. the work presents optimal algorithms in the settings where the buyer's valuations are fixed, stochastic and adversarial, respectively.
\cite{amin2013learning} also concerns 
selling through a posted price mechanism, but 
to a strategic buyer who may choose not to accept a price bellow her valuation (or accept a price above her valuation). 
The work presents learning algorithms in both the fixed valuation and stochastic valuation settings under the  assumption that discount their utilities over time. Other related works include \cite{golrezaei2020dynamic} which studies the dynamic pricing problem for repeated contextual second price auctions facing multiple strategic buyers. The work proposes learning algorithms that are robust to buyers' strategic behavior under various seller information structures and provides corresponding performance guarantees. \cite{golrezaei2019incentive} relaxes several assumptions for one of the settings in \cite{golrezaei2020dynamic}, and presents an algorithm with improved performance guarantees.
Finally, \cite{balseiro2019futility} considers the dynamic mechanism design problem against strategic buyers, and further identifies a class of problems in which the optimal mechanism is to simply repeat some static mechanism over time. The closes previous work to this paper is \cite{braverman2018selling}, where it studies the pricing problem against a single unconstrained quasi-linear buyer who adopts a certain class of learning algorithms, which they refer to as ``mean-based'' algorithms (e.g. Follow the Perturbed Leader algorithm and EXP3), the seller can extract the buyer's entire surplus; see \cite{deng2019prior} for an extension of this work. We remark that all works discussed here do not consider constrained buyers, and therefore this paper distinguishes itself by studying the pricing problem against buyers with both budget and ROI constraints, which further allows us to characterize special structures of seller revenue (see Section \ref{sec:reform}).

We refer readers to Appendix \ref{app:relatedwork} for an extended literature review.

\section{Preliminaries}
\label{sec:model}

\textbf{Notation.}
Let $\R_{+}$ be all non-negative  real numbers, and $\R_{++}$ be all strictly positive real numbers. For integer $N \in \N$, denote $[N] = \{1,2, \dots , N\}$ and $\Delta_{N} = \left\{\bm{p}\in [0,1]^{N}:\sum_{n \in [N]} p_{n} = 1\right\}$ be the $N$-dimensional probability simplex. Finally, denote $\norm{\cdot}$ as the Euclidean norm.

\textbf{Model setup: } Consider a seller repeatedly selling items to a buyer over $T$ periods through a posted price mechanism: in each period $t$, the seller posts  a price $d_{t}$ for the item to be sold, and the buyer makes a take it or leave it decision $z_{t}\in \{0,1\}$ based on her value $v_{t}$ of the item, where $z_{t} =1$ when the buyer takes the item at price $d_{t}$, and 0 otherwise.

We assume the seller commits to a finite price set $\mathcal{D} = \{D_{m}\}_{m\in[M]}$ where $1\geq D_{1} > \dots > D_{M} > 0$ from which she chooses the posted prices  $\{d_{t}\}_{t\in [T]}$, and we assume the the buyer's valuations are drawn independently each period from a distribution over $\bm{g} = (g_{1}\ldots g_{N}) \in \Delta_{N}$  ($g_{n} \in \R_{++}$ for all $n \in [N]$)
over a finite support $\mathcal{\V} = \{\V_{n}\}_{n\in[N]}$ where $1\geq \V_{1} > \dots > \V_{N} > 0$  such that $\prob(v_{t} = \V_{n}) = \valpe_{n}$ for any period $t \in [T]$.

\textbf{ROI and budget constrained buyers: }
The buyer aims to maximize total acquired value over $T$ periods, while subject to an ROI constraint with the target ROI of $\gamma\geq 1$ and a budget constraint with  budget rate $\rho \in (0,1)$.\footnote{Note that in the literature another common buyer objective is to optimize  linear utility that takes the form $\sum_{t\in [T]}(v_{t} - \alpha d_{t})z_{t}$ for some parameter $\alpha \geq 0$. We point out that all results in this paper can be extended easily to such linear objectives. } Mathematically, using the shorthand notation $\bm{d}_{1:T}$ for the sequence of prices $\{d_{t}\}_{t\in [T]}$, the buyer's hindsight optimization problem can be written as followed
\begin{align}
\label{eq:defOpt}
\begin{aligned}
 \opt(\bm{d}_{1:T}) =
\max_{\bm{z} \in [0,1]^{T}}  & \quad  \sum_{t\in[T]} \E\left[v_{t} z_{t} \right] \quad 
\textrm{s.t.} ~ &   \sum_{t\in [T]} \E\left[\left(v_{t} - \gamma d_{t}\right)z_{t}\right]\geq 0 \text{, and } \sum_{t\in [T]}  \E \left[ d_{t}z_{t} \right]\leq \rho T\,.
\end{aligned}
\end{align}
We remark that both budget and ROI constraints are studied in expectation. Such ``soft'' constraints are useful in practice due to the fact that real-world advertisers 
typically engage in many different online advertising campaigns, so it is reasonable to maintain these financial constraints in an average sense. We note that such soft financial constraints are also studied in mechanism design and online learning literature such as \cite{vaze2018online,golrezaei2018auction}.

We denote the optimal hindsight buyer decision sequence to Equation (\ref{eq:defOpt}) as $\{z_{t}^{*}(\bm{d}_{1:T})\}_{t\in [T]}$. When all prices are equal, i.e. $d_{t} = d$ for all $t$, we use the shorthand notation 
$\opt(d) $ and $\{z_{t}^{*}(d)\}_{t\in [T]}$. Note that optimal hindsight decisions $\{z_{t}^{*}(\bm{d}_{1:T})\}_{t\in [T]}$ may possibly be fractional, which can be implemented by randomization. 

The buyer's target ROI $\gamma$ and budget rate $\rho$ are private to the buyer and unknown to the seller. Also, both the seller and the buyer do not know the valuation distribution $\valp$.

\textbf{Seller's benchmark revenue and regret.}

The seller does not know the buyer's model primitives, namely the buyer's valuation distribution $\valp$, target ROI $\gamma$ and budget rate $\rho$. Furthermore, the seller only observes the buyer's decision $z_{t} \in \{0,1\}$, and does not observe buyer values. Under such information structure, we focus on non-anticipative seller pricing strategies that post prices based on historical data, i.e. in each period $t$, the decision $z_{t}$ can only depend on $\{(d_{\tau},z_{\tau})\}_{\tau\in [t-1]}$. We evaluate the performance of any sequence of pricing decision $\{d_{t}\}_{t \in [T]} \in \mathcal{D}^{T}$ by benchmarking its realized revenue, namely $\sum_{t\in[T]}d_{t}z_{t}$, to the maximum  revenue that could have been obtained if (i) the seller had set a fixed price over all $T$ periods and (ii) the buyer makes optimal hindsight decisions given her ROI and budget constraints. Mathematically, assume the seller fixes price $d \in \mathcal{D}$ over all $T$ periods, and the buyer's optimal decisions are $\{z_{t}^{*}(d)\}_{t\in[T]}$. Then, the seller's benchmark revenue is $\max_{d\in \mathcal{D}} \E[ d \sum_{t\in[T]} z_{t}^{*}(d)]$ and her regret
 can be defined as follows
\begin{align}
\label{eq:sellerregret}
   \sellreg= \max_{d\in \mathcal{D}}  d \sum_{t\in[T]} \E\left[z_{t}^{*}(d)\right]- \sum_{t\in [T]}\E\left[d_{t}z_{t}\right]\,,
\end{align}
where the expectation is taken w.r.t. $\{v_{t}\}_{t\in[T]}$  and randomness in the buyer's strategy (and thus randomness in $\{z_{t}^{*}(d)\}_{t\in[T]}$).

\begin{remark}
 The seller's regret resembles that of an $M$-arm multi-arm bandit (MAB) problem (see \cite{lattimore2020bandit} for a detailed introduction), where we can view each price $d \in \mathcal{D}$ as an arm and $d \cdot z_{t}$ as the reward by pulling arm $m$. Nevertheless, we point out that our setting is more complex than the vanilla MAB setting  as the seller's reward $d\cdot z_{t}$ for setting price $d$ during period $t$  not only depends on the seller algorithm which determines prices based on historical observations , but also the buyer's algorithm to optimize Equation (\ref{eq:defOpt}).
\end{remark}

We point out that the benchmark revenue in the seller's regret of Equation (\ref{eq:sellerregret}) is strong, as it represents the maximum seller revenue when both the buyer and seller have complete information and act optimally, i.e. if the seller knows everything about the buyer, in each period she myopically posts a revenue-maximizing price under best buyer response. 

Our goal is to develop a seller pricing algorithm to minimize regret when facing a buyer who optimizes Equation (\ref{eq:defOpt}) via running some online learning algorithm (to be discussed in later sections).

\section{Seller's Revenue and Regret}
\label{sec:reform}

In this section, we present a reformulation for the  seller's benchmark revenue in the seller's regret (Equation (\ref{eq:sellerregret})), and then further characterize  special structures of this reformulation which will later motivate the design of our pricing algorithm.
\subsection{Reformulating the seller's benchmark  revenue}
Recall the seller's benchmark revenue in Equation (\ref{eq:sellerregret}) which depends on the buyer's best response decision sequence over the entire horizon $T$ under a fixed price. To present our reformulation of this benchmark, we first show that for any price $d$, although the buyer's hindsight optimal decisions $\{z_{t}^{*}(d)\}_{t\in [T]}$ may seemingly be interdependent across periods due to the coupling of budget and ROI constraints over the entire horizon, the optimal buyer decision in each period $t$ simply requires the buyer to myopically make a decision $z_{t}$ that maximizes single-period expected value under ``single-period budget and ROI constraints'', namely $\E\left[\left(v_{t} - \gamma d\right)z_{t}\right]\geq 0$ and $ \E \left[d z_{t} \right]\leq \rho$.

Formally, consider the following myopic buyer optimization problem: for a given posted price $d$, let $\bm{x}\in [0,1]^{N}$ be some vector whose $n$th entry $x_{n}$ denotes the probability of accepting the price when the buyer's realized value is $V_{n}$. Then, the myopic buyer optimization problem can be written as Equation (\ref{eq:pricing:buyerutil}) whose optimal solution is shown in the following Lemma \ref{lem:thresholdOptSol} (see proof in Appendix \ref{pf:lem:thresholdOptSol}).
\begin{align}
\begin{aligned}
\label{eq:pricing:buyerutil}
\util(d) =  \max_{\bm{x}\in [0,1]^{N}}  \sum_{n\in[N]} \valpe_{n}\V_{n} x_{n} ~
 \quad \textrm{s.t.} &\sum_{n\in[N]} \valpe_{n}\left(\V_{n} - \gamma d\right)x_{n}\geq 0, ~ \text{ and } ~d\sum_{n\in[N]} \valpe_{n} x_{n}\leq \rho \,.
 \end{aligned}
\end{align}
\begin{lemma}
\label{lem:thresholdOptSol}
For any price $d$, the optimal solution to Equation \eqref{eq:pricing:buyerutil} is unique, and takes the form $\bm{x}_{d} = (1,1, \ldots q, 0,0\ldots 0) \in [0,1]^{N}$ for some $q \in (0,1]$.
\end{lemma}

The special form of the optimal solution of Equation (\ref{eq:pricing:buyerutil}) suggests a buyer strategy that accepts all items when buyer value is beyond a certain threshold. We formalize such a strategy in the following definition.
\begin{definition}[Threshold strategy]
\label{def:threshstrat}
For a given vector $\bm{x}$ that takes the form $\bm{x} = (1,1, \ldots q, 0,0\ldots 0) \in [0,1]^{N}$
where $q \in (0,1]$ is the $n$th entry, we say a buyer adopts a threshold strategy w.r.t. $\bm{x}$ if, regardless of the posted price, she accepts the item when her value is $V_{1} \ldots V_{n-1}$; accepts w.p. $q$ when her value is $V_{n}$; and rejects the item otherwise.
\end{definition}
As an example, for $N = 4$ and some vector $\bm{x} = (1,1,0.3, 0)$, the buyer adopts a threshold strategy w.r.t. $\bm{x}$ if she accepts the item when her value is $V_{1}$ or $V_{2}$; accepts w.p. 0.3 when her value is $V_{3}$, and rejects when her value is $V_{4}$.

With Lemma \ref{lem:thresholdOptSol} and the notion of threshold strategies in Definition \ref{def:threshstrat}, we can formally define the buyer's best response to a given price $d$:
\begin{definition}[Buyer best response]
\label{def:br}
We say a buyer best responds to a posted price $d$ if she adopts a threshold strategy w.r.t. $\bm{x}_{d} \in [0,1]^{N}$ which is the optimal solution to $U(d)$ (see Lemma \ref{lem:thresholdOptSol}).
\end{definition}
Note that in order for to best respond to a posted price, the buyer would need to know the value distribution $\bm{g}$.

Our main result for this subsection is illustrated in the following theorem, which states that buyer's hindsight optimal decision sequence $\{z_{t}^{*}(d)\}_{t\in [T]}$ for $\opt(d)$ in Equation (\ref{eq:defOpt}) simply requires the buyer to independently best respond to the posted price in each period.
\begin{proposition}
\label{prop:reformsellregret}
Given a single price $d$ posted across all periods,  the optimal buyer decision in each period $t$ is to best respond accorind a threshold strategy w.r.t. $\bm{x}_{d}$ (Definition \ref{def:br}), where $\bm{x}_{d} \in [0,1]^{N}$ is the unique optimal threshold solution to $\util(d)$ (Equation (\ref{eq:pricing:buyerutil})). Further, the best response buyer decision induces a per-period expected revenue
\begin{align}
\label{eq:revenue}
\rev(d) \defeq d\sum_{n\in [N]} \valpe_{n} x_{d,n}\,.
\end{align}
Then, $ \max_{d\in \mathcal{D}} \E\left[ d \sum_{t\in[T]} z_{t}^{*}(d)\right] = T\max_{d\in \mathcal{D}}\rev(d)$
and thus $ \sellreg = T\max_{d\in \mathcal{D}}\rev(d) - \sum_{t\in [T]}\E\left[d_{t}z_{t}\right]$.
\end{proposition}
We refer readers to the proof in Appendix \ref{pf:prop:reformsellregret}.

\subsection{Structure of Benchmark Seller Revenue}
Here, we present a special  underlying  structure of the seller revenue $\rev(d)$ defined in Equation (\ref{eq:revenue}) which will  motivate our pricing algorithm in the next Section \ref{sec:pricingalgo}. The goal of this section is to develop efficient ways to identify $\arg\max_{d\in\mathcal{D}}\rev(d)$ by avoiding exploring each possible price in $\mathcal{D}$. In the rest of the paper, we make the following assumption to rule out trivial problem instances (e.g. cases when the optimal solution $\bm{x}_{d}$ corresponding to some $d\in \mathcal {D}$ has all 0 entries or when one of the constraints are redundant):
\begin{assumption}
\label{assum:price:nontriv}
For any $d \in \mathcal{D}$, assume $V_{N} - \gamma d < 0< \V_{1} - \gamma d $ and $\sum_{n \in [N]}(\V_{n} - \gamma d)\valpe_{n} \neq 0$. Furthermore, assume $D_{M} < \rho < D_{1}$.
\end{assumption}

 To begin with, we  categorize all prices  $d\in \mathcal{D}$ based on whether constraints are binding under the corresponding optimal solution $\bm{x}_{d}$. 
\begin{definition}For price $d$ let $\bm{x}_{d}$ be the optimal threshold-based solution to $\util(d)$ in Equation (\ref{eq:pricing:buyerutil}). Then we call $d$
\begin{itemize}
    \item \textbf{Non-binding}, if under $\bm{x}_{d}$, both constraints are non binding, i.e., $d\sum_{n \in [N]}\valpe_{n}x_{d,n} < \rho$ and 
    $\sum_{n\in[N]}\left(\V_{n} - \gamma d\right)\valpe_{n}x_{d,n} > 0$;
    \item \textbf{Budget binding} if under $\bm{x}_{d}$, the budget constraints is binding, i.e.
    $d\sum_{n \in [N]}\valpe_{n}x_{d,n} = \rho$ and $\sum_{n \in [N]}(\V_{n} - \gamma d)\valpe_{n}x_{d,n} > 0$;
    \item \textbf{ROI binding}  if under $\bm{x}_{d}$, the ROI constraint is binding, i.e. $\sum_{n \in [N]}(\V_{k} - \gamma d)\valpe_{n}x_{d,n} = 0$ and $d\sum_{n \in [N]}\valpe_{n}x_{d,n}\leq \rho$.
\end{itemize}
\end{definition}
It is apparent that any price $d\in \mathcal{D}$ must belong to at least one of these categories. Also, if a price is non-binding, it cannot be budget binding or ROI binding. 

Our main result of this subsection is the following Theorem \ref{thm:pricing:optd}, which states that as we traverse $\mathcal{D}$ in increasing price order, prices are first non-binding and the revenue $\rev(d)$ increases in $d$; then prices become budget binding, where revenue remains constant at $\rev(d) = \rho$; finally prices become ROI binding, where $\rev(d)$ decreases in $d$. The proof can be found in Appendix \ref{pf:thm:pricing:optd}.
\begin{theorem}[Bell-shaped Structure of the Revenue Function]
\label{thm:pricing:optd}
Suppose that Assumption $\ref{assum:price:nontriv}$ holds. Then, the following hold  
\begin{enumerate}
    \item For any non-binding prices $d,\widetilde{d}$, if $d< \widetilde{d}$ then $\rev(d) < \rev(\widetilde{d})$. 
    \item If $d$ is budget binding, any price $\widetilde{d} > d$ cannot be non-binding, which means $\widetilde{d}$ is budget binding or ROI binding.
    \item If $d$ is ROI binding, then any $\widetilde{d} > d$  must also be ROI binding. Furthermore, $\rev(d) > \rev(\widetilde{d})$. 
\end{enumerate}
\end{theorem}
We provide an illustration of Theorem \ref{thm:pricing:optd} in Figure \ref{fig:transition} that depicts the ``non-binding $\rightarrow$ budget binding $\rightarrow$ ROI binding''  transition phenomenon, as well as a corresponding revenue ``increase $\rightarrow$ plateau $\rightarrow$ decrease'', as we traverse prices in increasing order. We note that for specific model primitives $\valp ,\gamma, \rho$, there may exist no budget binding prices
(as shown in right subfigure in Figure \ref{fig:transition}), meaning that there are scenarios in which it is impossible for the buyer to extract the entire buyer budget. Nevertheless, this transition phenomena suggests that we can efficiently identify the maximizing revenue $\arg\max_{d\in\mathcal{D}}\rev(d)$ by utilizing a simple binary search approach. Hence, we utilize this structure of $\rev(d)$ to motivate our pricing algorithm. 
\begin{figure}[!ht]
\includegraphics[width = \linewidth]{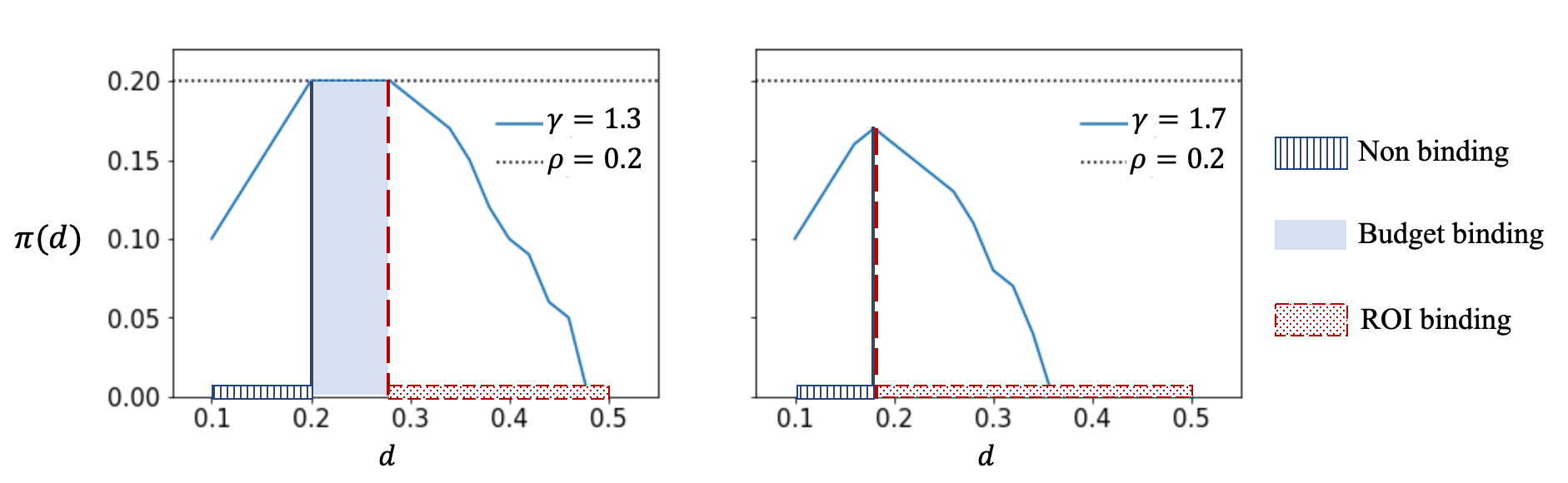}
\centering
\caption{\footnotesize \textbf{Seller revenue function bell-shape structure.} Model primitives: number of unique buyer valuations $N = 6$, valuation set $\mathcal{V} = (0.6,0.5,0.4,0.3,0.2,0.1)$, valuation distribution $\valp = (0.1,0.1,0.2,0.1,0.2,0.3)$, seller price set $\mathcal{D} = (0.5,0.48 \dots 0.1)$, buyer budget rate $\rho = 0.2$. The left and right subfigures correspond to target ROI $\gamma = 1.3$ and $1.7$ respectively. In both cases, prices transition from non-binding to budget binding, and finally to ROI binnding. Revenue $\rev(d)$ increases as in $d$ when prices are non-binding, decreases in $d$ when prices are ROI binding, and remains at $\rho$ when prices are budget binding. Note that when $\gamma = 1.7$, there are no budget binding prices.}  
\label{fig:transition}
\end{figure}

\section{Pricing Algorithm against an ROI and Budget Constrained Buyer}
\label{sec:pricingalgo}
The main challenge the seller faces is her lack of knowledge on the buyer's model primitives, namely the buyer's valuation distribution $\valp$, target ROI $\gamma$ and budget rate $\rho$. Furthermore, the seller has limited information feedback as she only observes whether the buyer accepted the price or not, i.e., the seller only observes the outcome $z_{t} \in \{0,1\}$. This lack of information makes it very difficult for the seller to estimate the buyer's model primitives. Nevertheless, we propose a simple pricing algorithm that bypasses this lack of knowledge via exploiting the price transition phenomenon as characterized in Theorem \ref{thm:pricing:optd} and Figure \ref{fig:transition}. We demonstrate later in subsection \ref{subsec:sellerregret} that this algorithm achieves good performance when facing a general class of algorithms that is adaptive to nonstationary environments.

Our proposed pricing algorithm consists of an exploration phase and an exploitation phase.  During the exploration phase, the algorithm searches for a revenue maximizing price $\arg\max_{d\in \mathcal{D}}\rev(d)$ through an episodic structure: the seller initiates the first episode  $\mathcal{E}_{1}$, and fixes the price chosen in this episode $D_{1}$ for $E$ consecutive periods. At the end of the episode (i.e. after $E$ periods since the beginning of the episode), the seller records the average per-period revenue $\Hat{\rev}(D_{1}) = \frac{D_{1}}{E}\sum_{t\in \mathcal{E}_{1}}z_{t}$, where $z_{t} \in \{0,1\}$ indicates whether the buyer takes the price at time $t \in \mathcal{E}_{1}$. The process then repeats as the seller moves on to episodes $\mathcal{E}_{2}, \dots $ This exploration phase eventually terminates when the seller has explored enough prices.
The seller's pricing decision in each episode is governed by a binary search procedure over the price set $\mathcal{D}$, such that
every price is chosen at most once across all episodes, and the exploration phase will have $\mathcal{O}(\log(M))$ episodes. Our pricing algorithm is detailed in 
Algorithm \ref{alg:pricing}. 

We note that our proposed algorithm does not try to learn the buyer's model primitives.
We further point out that 
such a binary-search approach is a natural choice to identify revenue-optimal prices in the simplest monopolistic pricing setting under a typical unimodal 
assumption, \footnote{In monopolistic pricing, the revenue-optimal price $p^{*}$ is charachterized by $d^{*}  = \arg\max_{d} dF(d)$, where $F$ is the cdf of buyer valuations. A typical assumption is such that the function $dF(d)$ is unimodal.} and one may wonder whether this approach can have good performances against a much more complex setting where the buyer is ROI and budget constrained and aims to learn her optimal bidding strategy. Surprisingly, in the next section we are in fact able to show this simple approach achieves good performances against buyers who are adaptive to price changes.

\begin{algorithm}[h]
    \centering
    \caption{Episodic Binary Search }\label{alg:pricing}
    \footnotesize
    \begin{algorithmic}[1]
    \Require Exploration episode length $E$.
    \State Initialize iteration index $\iter= 1$.
    {\Statex{\textbf{Exploration episodes}}}:
    \State Set $D_{1}$ for $E$ consecutive periods, and record per-period revenue $\Hat{\rev}\left(\D_{1}\right)$. Then 
	set $\D_{M}$ for $E$ consecutive periods, and record average per-period revenue $\Hat{\rev}\left(\D_{M}\right)$.
	\State Set $\term\leftarrow \arg\max_{m\in \{1,M\}} \Hat{\rev}\left(\D_{m}\right)$\; 
	$\lef= 1$, $\rig = M$, $\med = \lfloor\frac{\lef+\rig}{2}\rfloor$.
	\While{$\lef < \rig$}
	\State $\iter \leftarrow \iter+1$.
	\If{per-period revenue $\Hat{\rev}\left(\D_{k}\right)$ is not recorded for $k = \med,\med+1$}
	\State{Set
		price $\D_{k}$ for $E$ consecutive periods and record per-period revenue $\Hat{\rev}\left(\D_{k}\right)$ for $k = \med, \med+1$}
	 \EndIf
	 \If{$\Hat{\rev}\left(\D_{\med}\right)  < \Hat{\rev}\left(\D_{\med+1}\right)$} 
	 \State{Set $\term\leftarrow \arg\max_{m\in \{\term,\med+1\}} \Hat{\rev}\left(\D_{m}\right)$, $\lef \leftarrow \med+1$, $\med \leftarrow \lfloor\frac{\lef+\rig}{2}\rfloor$}
	 \Else 
	 \State Set $\term\leftarrow \arg\max_{m\in \{\term,\med\}} \Hat{\rev}\left(\D_{m}\right)$, $\rig \leftarrow \med-1$, $\med \leftarrow \lfloor\frac{\lef+\rig}{2}\rfloor$
	 \EndIf
	\EndWhile

	{\Statex{\textbf{Exploitation episode}}}:
	\State{Set price $\D_{\term}$ for the remaining periods.}
    \end{algorithmic}
\end{algorithm}
For notation convenience, we denote $\mathcal{E}_{h}$ as the collection of periods in episode $h$. Finally, we remark that the exploration episode length $E$ is deterministic and depends on the total number of periods $T$.

\subsection{Regret Analysis of Pricing Algorithm}
\label{subsec:sellerregret}
In this section, we provide theoretical guarantees for our proposed pricing algorithm against  buyer algorithms whose induced decisions approximate single-round best responses (see Definition \ref{def:br}) in the average sense. We formally define algorithms with such properties as follows:
\begin{definition}[$\xi$-Adaptive Buyer Algorithms]
\label{def:adapt}
We say a buyer algorithm is $\xi$-adaptive to seller algorithm \ref{alg:pricing} if the induced decisions $\{z_{t}\}_{t\in [T]}$ in any exploration or exploitation episode $\mathcal{E}_{h}$ satisfies
\begin{align}
\label{eq:adapterror}
    \left| \frac{D_{h}}{|\mathcal{E}_{h}|} \sum_{t\in \mathcal{E}_{h}}z_{t} - \rev(D_{h})\right| \leq \frac{\phi (|\mathcal{E}_{h}|)}{|\mathcal{E}_{h}|}
\end{align}
with probability (w.p.) at least $1-1/T$ for some increasing error function $\phi: \R_{+} \to \R_{+}$ and $\phi (x) =  \mathcal{O}(x^{1-\xi}) $. Here $D_{h}$ is the price set in episode $h$, and $\rev(\cdot)$ is the per-period revenue function under buyer best response defined in Equation (\ref{eq:revenue}).
\end{definition}
The term $\left|\frac{D_{h}}{|\mathcal{E}_{h}|}\sum_{t\in \mathcal{E}_{h}}z_{t} - \rev(D_{h})\right|$  is the seller's average revenue loss, relative to the revenue from optimal buyers, over a certain period with a  fixed price $D_{h}$. Alternatively, the term can  be viewed as the buyer's deviation from best responding since $\frac{\rev(D_{h})}{D_{h}} = \sum_{n \in [N]} \valpe_{n}x_{D_{h},n}$ is the optimal probability with which the buyer should take price $D_{h}$.

The main result of this subsection is presented in Theorem \ref{thm:episodicprice}, which characterizes the performance of our pricing algorithm against any $\xi$-adaptive buyer algorithm. The proof of Theorem \ref{thm:episodicprice} can be found in Appendix \ref{proof:thm:episodicprice}.
\begin{theorem}[Pricing against $\xi$-adaptive buyers]
\label{thm:episodicprice}
{Consider the seller runs Algorithm \ref{alg:pricing} against an $\xi$-adaptive buyer
algorithm (Definition \ref{def:adapt}).} Fix $\epsilon \in (0,\xi)$ independent of $T$. Then by setting exploration episode length $E = T^{1-\xi + \epsilon}$ in seller algorithm \ref{alg:pricing}, for large enough $T$ under Assumption \ref{assum:price:nontriv} the seller's regret is bounded as
\begin{align}
\begin{aligned}
\label{eq:sellregret}
   \sellreg \leq & 2\left(\left\lfloor\log_{2}(M) \right\rfloor + 1\right) \cdot T^{1-\xi + \epsilon}  + \phi\left(T\right)+ \left(\left\lfloor\log_{2}(M)\right\rfloor + 1\right)^{2}/2\,,
  \end{aligned}
\end{align}
{where $\phi$ is the error function defined Equation (\ref{eq:adapterror}).}
\end{theorem} 

The first term $ T^{1-\xi + \epsilon}$ in the seller's regret (see Equation (\ref{eq:sellerregret}))  characterizes the number of periods required for the buyer's algorithm to  approximate the best-responding decisions in each episode facing a fixed price;  the second term $ \phi\left(T\right) $  represents the buyer's deviation from the best response. Finally, we point out that although in Theorem \ref{eq:sellregret} we set the exploration episode length to be $E = T^{1-\xi + \epsilon}$, the seller does not need to know the exact value of $\xi$ as a lower bound would be sufficient: if the seller knows some lower bound for $\xi$, say   $\xi' < \xi$, she can set $E = T^{1-\xi'}$, and the final seller regret would become $\sellreg \leq 2\left(\left\lfloor\log_{2}(M) \right\rfloor + 1\right) \cdot T^{1-\xi'}  + \phi\left(T\right) + \left(\left\lfloor\log_{2}(M)\right\rfloor + 1\right)^{2}/2$ for large enough $T$.

Another interesting observation for the seller regret is that its dependence on the price set dimension $M$ is logarithmic, meaning that our Algorithm \ref{alg:pricing} is  robust  w.r.t. the size of the seller's decision set. In fact, later in Section \ref{sec:adddisc}, we discuss that this nice logarithmic dependence on $M$ allows us to easily handle continuous price sets without causing decay in seller performance by using a simple discretization approach.

\section{Example of Adaptive and Buyer-regret minimizing Algorithms}
\label{sec:buyeralgo}
In this section, we present simple examples of buyer algorithms that are  adaptive in the sense of Definition \ref{def:adapt}, and also aim to satisfy  budget and ROI constraints (Equation (\ref{eq:defOpt})) while attaining low buyer regret, where the regret of the buyer is defined as
\begin{align}
\label{eq:buyerregret}
  \buyreg= \opt(\bm{d}_{1:T}) - \sum_{t\in [T]}\E\left[v_{t}z_{t}\right] \,.
\end{align}
Here $\{z_{t}\}_{t\in[T]}$ is the sequence of buyer binary decisions produced by the buyer algorithm. Also recall $\opt$ is the buyer's optimal hindsight total value described in Equation (\ref{eq:defOpt}). In the following subsections, we consider a clairvoyant buyer who best responds in each period as well as a buyer who possess machine-learned (ML) advice with which she uses to make decisions. We then further characterize seller regret of our proposed Algorithm \ref{alg:pricing} against such buyers.

\subsection{Best-responding buyer}
\label{subsec:br}
As a warm-up buyer example, we first consider a clairvoyant buyer who knows her value distribution $\bm{g}$, which means the buyer has nothing to learn from the data and thus can best respond in the sense of Definition \ref{def:br} during each period to maximize value under both budget and ROI constraints (Equation (\ref{eq:defOpt}). We show in the following lemma that best responding is adaptive (see proof in Appendix \ref{pf:lem:bestRespondAdapt}).
\begin{lemma}[Best-responding is 1/2-adaptive]
\label{lem:bestRespondAdapt}
There exists some $T_{0} \in \N$ such that for all $T > T_{0}$, best responding is $\frac{1}{2}$-adaptive (Definition \ref{def:adapt}). 
\end{lemma}
Combining Lemma \ref{lem:bestRespondAdapt} and Theorem \ref{thm:episodicprice}, we present the regret of Algorithm \ref{alg:pricing} against a best responding buyer in the following theorem whose proof can be found in Appendix \ref{pf:corr:bestRespondBuyer}
\begin{theorem}[Seller's regret against best responding buyer]
\label{corr:bestRespondBuyer}
 Assume the buyer always best responds, then for a fixed $\epsilon \in (0,\frac{1}{2})$ independent of  $T$, if  the seller sets prices with episode length $E = T^{\frac{1}{2}+\epsilon}$
using Algorithm \ref{alg:pricing}, then for large enough $T$, the seller's regret is bounded as $\text{Reg}_{\text{sell}} \leq \Theta(T^{\frac{1}{2}+\epsilon})$. On the other hand, the buyer also incurs $\Theta(T^{\frac{1}{2}+\epsilon})$ regret, and both budget and ROI constraints are satisfied. 
\end{theorem}
In this clairvoyant buyer setting, since the buyer is not learning and always best responds, the $T^{\frac{1}{2}}$ constituent in the seller regret is due to learning error from the seller. In the next section, we introduce a buyer who is non-clairvoyant and also constantly learns how to respond, and further discuss how buyer and seller learning errors simultaneously impact seller regret.

\subsection{Buyer with machine-learned (ML) advice}
In a real world scenario, buyers typically do not know their value distribution; e.g. buyers may be unaware of the likelihood of conversion of their ad impressions. However, the emergence of data-driven tools for online advertising platforms have provided buyers with additional analytics, or so-called ML advice, to help buyers estimate ad conversion. In this subsection, we consider a buyer who possesses ML advice in the form of distribution estimates of $\bm{g}$ with which she uses to approximate best responses against posted prices. Formally, we characterize such ML-advice-driven buyer responses as followed:

\begin{definition}[Approximate best response with ML advice.] 
\label{def:approxbr}
Assume in each period $t$, the buyer obtains ML advice $\Hat{\bm{g}}_{t}\in \Delta_{N}$ that only depends on historical data $\{v_{\tau}\}_{\tau\in[t]}$ s.t. $\norm{\Hat{\bm{g}}_{t} - \bm{g}} < \ell_{t}$ where $\ell_{t}$ is some estimation error. Then, the buyer solves for the optimal solution $\Hat{\bm{x}}_{t}$ in Equation (\ref{eq:pricing:buyerutil}) via replacing the true distribution $\bm{g}$ with $\Hat{\bm{g}}_{t}$, and then adopts a threshold strategy w.r.t. $\Hat{\bm{x}}_{t}$ (see Definition \ref{def:threshstrat}).
\end{definition}
We remark that ML advice in the form of distributional estimates is very common. For model-based approaches, ML algorithms assume distributions take a certain parametric form and then uses data to estimate unknown distribution parameters; see e.g. \cite{eliason1993maximum} for an intro on maximum likelihood estimation. For more general non-parametric approaches, ML advice concerns using empirical estimates (or so-called histogram estimates), which we will later discuss in Theorem \ref{thm:sellregretMLee}.

The following lemma relates ML advice driven approximate responses to our notion of buyer adaptivity in Definition \ref{def:adapt}, with which we are able to quantify seller regret in light of Theorem \ref{thm:episodicprice}. The detailed proof can be found in Appendix \ref{pf:thm:sellregretML}

\begin{theorem}[Seller regret against approximate best responding buyer with ML advice]
\label{thm:sellregretML}
Assume the buyer approximate best responds with ML advice (Definition \ref{def:approxbr}) and there exists some $\tote \in (0,1)$ s.t. in each exploration or exploitation episode $h$ of Algorithm \ref{alg:pricing} the estimation errors, denoted by $\ell_t$'s,  satisfy $\lim_{t\to \infty}\ell_{t} = 0$ and $\sum_{t\in \mathcal{E}_{h}} \ell_{t} \leq \widetilde{\phi}(|\mathcal{E}_{h}|)$ for some increasing function $\widetilde{\phi}:\R_{+}\to \R^{+}$ and $\widetilde{\phi}(x) \leq \mathcal{O}(x^{1-\tote})$. Then this buyer algorithm is $\xi$-adaptive for $\xi = \min\{\frac{1}{2},\tote\}$. Further, by setting  exploration episode length $E = T^{1-\xi + \epsilon}$  for some $\epsilon \in (0,\xi)$ independent of $T$, the seller regret is exactly that in Equation (\ref{eq:sellregret}) of Theorem \ref{thm:episodicprice} for large enough $T$. On the buyer side, we have
$ \buyreg \leq \mathcal{O}(T^{1-\xi})$ and the induced buyer decisions $\{z_{t}\}_{t\in [T]}$ satisfy 
\begin{align*}
    \frac{1}{T}  \sum_{t\in [T]} \E\left[\left(v_{t} - \gamma d_{t}\right)z_{t}\right]\geq - \Theta(T^{-\tote}) ~\text{ and }~
     \frac{1}{T} \sum_{t\in [T]}  \E\left[d_{t}z_{t} \right]\leq \rho + \Theta(T^{-\tote})\,.
 \end{align*}
\end{theorem}
We remark that best responding buyers considered in Section \ref{subsec:br} can be viewed as a special case of buyers with ML advice where the advice is perfect, i.e. $\ell_{t} = 0$ for all $t$ so $\widetilde{\phi}(x) \equiv 0$ and consequently $\tote = 1$. This recovers our results in Theorem \ref{corr:bestRespondBuyer}.

Here, we also quickly discuss the aggregate impact of buyer and seller learning error on the seller regret of our proposed Algorithm \ref{alg:pricing}. In particular, the constituent $T^{1-\xi} = T^{1-\min\{\frac{1}{2},\tote\}}$ in the seller regret arises from learning errors of both the buyer and the seller. We can view the seller's learning rate to be in the order of $t^{-\frac{1}{2}}$, and the buyer learning rate to be of order $t^{-\tote}$, and thus we see that the seller regret is governed by the agent that learns at a slower rate: if the buyer is learning more slowly, i.e. $\tote < \frac{1}{2}$, then the seller regret is driven by the buyer learning loss; a similar argument applies for the case when the buyer learns more quickly.

To conclude this section, we present a concrete example for buyers with ML advice: consider the simple ML advice that is an empirical estimate of the buyer's value distribution: 
\begin{align}
\label{eq:empiricalestimates}
\Hat{\bm{g}}_{t} = \frac{1}{t} \cdot (\sum_{\tau \in [t]}\I\{v_{t} = V^{1}\} \ldots \sum_{\tau \in [t]}\I\{v_{t} = V^{N}\})\,.
\end{align}
Then, both the buyer and seller regret are  characterized in the following theorem  (see proof in Appendix \ref{pf:thm:sellregretMLee}).

\begin{theorem}[Seller regret against  buyer with empirical distribution estimates]
\label{thm:sellregretMLee}
When the buyer approximate best responds with ML advice in the form of empirical estimates as defined in Equation (\ref{eq:empiricalestimates}), Theorem \ref{thm:sellregretML} holds for $\tote= \xi = \frac{1}{2}$ w.p. at least $1-1/T$.
\end{theorem}

\section{Additional Discussions}
\label{sec:adddisc}
\textbf{Continuous price set.} We remark that our main results in this paper, specifically the analyses of Algorithm \ref{alg:pricing} and the corresponding seller regret, can be easily extended to handle continuous seller price sets, as the seller regret in Theorem \ref{thm:episodicprice} only depends logarithmically on $M$ which we recall to be the size of a discrete price set. Assuming the price decision set  is $[0,1]$, the approach that the seller can take is to discretize the decision set into  $\mathcal{D} = \{\frac{1}{T},\frac{2}{T} \ldots 1\}$ with size $|\mathcal{D}| = T$. Recall $\pi(d)$ defined in Equation (\ref{eq:revenue}) is the expected per-period seller revenue under buyer best response, and define $d^{*}=\arg\max_{d\in [0,1]} \rev(d)$ to be the optimal price w.r.t. the continuous set, such that the seller regret is now $\sellreg = T \cdot  \pi(d^{*}) - \sum_{t\in [T]}\E\left[d_{t}z_{t}\right]$ (see Proposition \ref{prop:reformsellregret}). Then, for a price $\widetilde{d}\in \mathcal{D}$ in the discretized set $\mathcal{D}$ that is close to $d^{*}$ such that $|\widetilde{d} - d^{*}| < \frac{1}{T}$, similar to our proof in Theorem \ref{thm:sellregretML} we can show that the optimal solutions $\bm{x}_{d}$ and $\bm{x}_{\widetilde{d}}$ to the per-period buyer optimization problem $U(d)$ and $U(\widetilde{d})$ (see Equation (\ref{eq:defOpt})), respectively, are also close to one another. Further, we can show that $\rev(d^{*}) - \rev(\widetilde{d})\leq O(\frac{1}{T})$. Therefore, via running Algorithm \ref{alg:pricing} w.r.t. the discretized price set $\mathcal{D}$, our seller regret when facing a $\xi$-adaptive buyer (Definition \ref{def:adapt}) can be bounded as
\begin{align*}
    & \sellreg 
    ~=~  T \max_{d\in [0,1]}\rev(d) - \sum_{t\in [T]}\E\left[d_{t}z_{t}\right] \\
    ~=~ & T \underbrace{( \pi(d^{*}) - \max_{d\in \mathcal{D}} \rev(d))}_{\text{discretization error}} + T \max_{d\in \mathcal{D}} \rev(d) -\sum_{t\in [T]}\E\left[d_{t}z_{t}\right]\\
    ~\leq~ & T  ( \pi(d^{*}) -  \rev(\widetilde{d})) + T \max_{d\in \mathcal{D}} \rev(d) -\sum_{t\in [T]}\E\left[d_{t}z_{t}\right]\\
    ~\leq~ & \mathcal{O}(1) + T \max_{d\in \mathcal{D}} \rev(d) -\sum_{t\in [T]}\E\left[d_{t}z_{t}\right]\\
    ~\leq~ & \mathcal{O}(1) + \mathcal{O}(\log(T)T^{1-\xi+\epsilon} + \phi(T))\,,
\end{align*}
where the final inequality follows from the seller regret (Equation (\ref{eq:sellregret})) in Theorem \ref{thm:episodicprice} by setting the price set size $M = T$. That being said, the discretization error introduced to the seller regret is only in the order of $\mathcal{O}(1)$, and this is due to the 
the fact that the
bell-shape structure of seller's revenue (Theorem \ref{thm:pricing:optd}) along with our seller algorithm yields a seller regret that is logarthmic in the discrete price set size.

\textbf{Ethics of buyer-seller interactions.}
As modern online ad platforms run selling mechanisms to sell ad impressions, they also offer services for buyers to help procure ad impressions on their behalf. This raises  potential ethical concerns regarding the issue of platforms controlling both the buyer algorithms and auction/pricing protocol. For example, the platform can sets high prices and simultaneously run buyer algorithms that would accept such prices, leading to large platform revenue margins while enforcing high costs to buyers. Nevertheless, in reality, procurement (on behalf of buyers) and pricing are either conducted by different and independent entities, or two non-collusive parties of the same entity where collaboration and any type of information flow that encourages collusion is prohibited. Our main goal of this paper is to shed light on the possible behavior and dynamics for online advertising markets under real financial considerations, and we believe that preventing such collusive behavior between buyer-seller interactions is a future research direction of practical and ethical importance.

\textbf{Other future directions.} One natural future research direction that is of both theoretical and practical interest involves designing pricing algorithms when facing multiple financially constrained buyers. The multi-buyer analogue to our single-buyer posted price setup in this work is to set a single reserve price in each period over time where constrained buyers compete in a second-price auction (see e.g. setup in \cite{golrezaei2019incentive} for non-constrained buyers). The key challenge lies in the fact that in this multi-buyer setup we no longer have the salient bell-shape structure in the seller revenue function, and more importantly buyer algorithmic interactions introduce significant difficulties to the analyses of seller regret. Similar challenges that arise from selling to multiple learning buyers have also been discussed (but not resolved) in related works such as \cite{braverman2018selling,deng2019prior}.

\bibliographystyle{informs2014}
\bibliography{ref}

\newpage
\begin{APPENDICES}
\begin{center}
\vspace{0.8cm}
    \Large Appendices for\\
    \vspace{0.2cm}
    \Large \textbf{Pricing against a Budget and ROI Constrained Buyer}
    \noindent\makebox[\linewidth]{\rule{1\linewidth}{0.9pt}}
\end{center}
\setcounter{page}{1}

\section{Extended Literature Review}
\label{app:relatedwork}

As the most closely related works have been discussed in the introduction section, here we only further discuss broader related works.

\paragraph{Other related work in online resource allocation}
 There has been extensive research on online resource allocation with budget/capacity constraints (see e.g. \cite{kleinberg2005multiple,devanur2009adwords,agrawal2016efficient}) and here we briefly discuss those that are the most relevant.\footnote{The buyer's online bidding problem can be viewed as an online resource allocation problem.  However, a key difference is that in bidding, the buyer does not observe the highest competing bid $d_{t}$ (equivalently the amount of resource depleted) before making a decision; as in the resource allocation problem, both the reward and resource depletion are revealed before decision making. Therefore, to apply a resource allocation algorithm in the bidding problem, one must additionally impose some bidding mechanic that indirectly achieves the desired allocation through constructing appropriate bid values. } \cite{zhou2008budget} studies the budget-constrained bidding problem for sponsored search in an adversarial setting and present an algorithm with competitive ratio that depends on upper and lower bounds on the value-to-cost ratios; \cite{babaioff2007knapsack,arlotto2019uniformly} study variants of the knapsack and secretary problems under the random order arrival model and stochastic arrival model, respectively, both presenting near optimal algorithms in their respective settings. Our work differs from this line of research as we incorporate an ROI constraint while also considering the problem of how to price against budget and ROI constrained buyers.
Finally, \cite{agrawal2014dynamic}  utilizes a primal-dual framework to study online linear programming (LP) with packing constraints, where the positive-valued constraint matrix is revealed column by column (each column corresponds to a highest competing bid $d_{t}$) along with the corresponding objective coefficient (corresponding to utility $v_{t} - \alpha d_{t}$). Their algorithm determines the decision variable corresponding to the arriving column based on the dual variables of past revealed columns.



\paragraph{Online bidding in repeated auctions under feedback constraints}
Other than budget capacities and ROI targets, buyers are also typically constrained in terms of the amount information available as they participate in auctions. For example, \cite{balseiro2019contextual} studies bidding problem in first price auctions under different feedback structures where an unconstrained quasi-linear  buyer only  observes whether or not she wins  the auction, and \cite{han2020optimal,han2020learning} study a similar problem  where the buyer also gets to observe the highest competing bid if she did not win the auction. 
As another related work, \cite{weed2016online} studies the bidding problem where 
the buyer does not know her valuation before submitting her bid, and only observes her valuation if she wins the auction. The work considers the stochastic and adversarial highest competing bid settings, and presents algorithms that build on the UCB and EXP3 algorithms, respectively.

\paragraph{Online optimization with covering constraints}
The buyer's ROI constraint takes the form of a long-term covering constraint. 
The related
problem of optimization under online covering constraints have been studied in   \cite{alon2003online,azar2013online,azar2014online}. However, the setting in these works differ from ours: Instead of making irrevocable online decisions, these works focus on updating a decision vector upon the arrival of a covering constraint each period such that this constraint is satisfied. In other words, they consider the decision problem where covering constraints are satisfied in each period, while our buyers of interest only need to satisfy the covering (ROI) constraint in the long run. Another key difference is that in these works the covering constraints are all positive, which means these constraints can be easily satisfied (per period) by increasing each entry of the decision vector. On  the contrary, in our problem the ROI balance per period $(v_{t} - \gamma d_{t})z_{t}$ may be negative, and hence makes constraint satisfaction more difficult.

\paragraph{Autobidding.} This paper is also related to a recent line of work that studies so-called ``autobidders'' who simultaneously participate in parallel auctions with the aim to maximize total value  subject to a global ROI and budget constraint, which says that total value accrued across auctions is no less than total spend times some multiple (i.e. the target ROI), and the total spend is less than a global budget. \cite{aggarwal2019autobidding} has first formulated the optimization problem for autobidders, and presented optimal bidding strategies for such bidders when all parallel auctions are truthful. \cite{deng2021towards, balseiro2021robust} study the price of anarchy when  multiple autobidders bid in parallel auctions of classic formats such as VCG, GSP and GFP. \cite{deng2022fairness} show auctioneers can set personalized reserve prices using predictions on bidder values (i.e. machine-learned advice) to improve welfare guarantees for individual bidders. 

\section{Additional definitions}
In this section, we introduce some additional definitions that will be used throughout the appendices.

\begin{definition}[Threshold vectors]
\label{def:thresholdvect}
We say that an $N$-dimensional vector $\bm{x}\in \R^{N}$ is a threshold vector if it takes the form of $\bm{x} =(1 \dots 1,  q, 0 \dots 0)$, where the first $J \in \{0,\dots N\}$ entries are 1's, followed by some number $q \in [0,1)$, and trailing with $(N-J-1)_{+}$ 0's.\footnote{For the edge case of $(1,\dots 1)\in \R^{N}$, $J =N$ and hence the number of trailing 0's is $(N-J-1)_{+} = 0$.} Any threshold vector is uniquely characterized by its dimension $N$, as well as, a tuple $(J,q) \in \{0,\dots N\}\times [0,1)$, so we denote the vector as  $\thr(J,q)$. In the special case when $J = N$, take $q = 0$.
\end{definition}

 For any two vectors $\bm{a},\bm{b} \in \R^{n}$, let $\min\{\bm{a},\bm{b}\} = (\min\{a_{i},b_{i}\})_{i\in[n]}$ be the element-wise minimum. We write $\bm{a} \preceq \bm{b}$ if and only if $a_{i} \leq b_{i}$ and  $\bm{a} \succeq \bm{b}$
 if and only if $a_{i} \geq b_{i}$ for all $i \in [n]$.

\section{Proofs for Section \ref{sec:reform}}
\label{app:reform}

\subsection{Proof of Lemma \ref{lem:thresholdOptSol}}
\label{pf:lem:thresholdOptSol}
Here, we show a more detailed version of the lemma stated as followed:

\begin{theorem}[Detailed version of Lemma \ref{lem:thresholdOptSol}]
\label{lem:thresholdOptSolext}
For a fixed price $d$, define
\begin{align}
\label{eq:threshold}
\begin{aligned}
    &   \jroi = \max\big\{n\in [N]: \sum_{\ell \in[n]} g_{\ell} \left(\V_{\ell} - \gamma d\right)\geq 0\big\}, \quad q_{\ROI} = \frac{\sum_{k \in [\jroi]}g_{n} \left(\V_{n} - \gamma d\right)}{g_{\jroi+1}\cdot\left|V_{\jroi+1} - \gamma d \right|}\,,\\
     &  \jbud= \max\big\{n\in [N]: d\sum_{\ell \in[n]} g_{\ell} \leq \rho\big\}, \quad \text{ and } \quad q_{\bud} = \frac{\rho - d \sum_{k \in [\jbud]} g_{n}}{g_{\jbud+1} \cdot d}\,,
    \end{aligned}
\end{align}
If we let $\bm{x}_{\ROI} =\thr(\jroi,q_{\ROI})$ and $\bm{x}_{\bud} =\thr(\jbud,q_{\bud})$ be two threshold vectors (see Definition \ref{def:thresholdvect}), then $\bm{x}_{d} =  \min\left\{\bm{x}_{\ROI},\bm{x}_{\bud}\right\}$
is the unique optimal solution to $\util(d)$ in Equation \eqref{eq:pricing:buyerutil}. Furthermore, $\bm{x}_{d}$ is also a threshold vector characterized by tuple $(J,q)$ where 
\begin{align} 
\label{eq:optremainderprob}
    \J =\min\{ \jroi, \jbud\}, \quad  q = x_{d,\J+1} =  \min\left\{x_{\bud,\J+1},x_{\ROI,\J+1}\right\}\,.
\end{align}
\end{theorem}

\textit{Proof.}

Our proof for Theorem  \ref{lem:thresholdOptSolext} consists of 3 steps:
\begin{itemize}
 \item \textbf{Step 1.} We show that $\bm{x}_{\bud}$ is the unique optimal solution to the ``budget constraint only'' problem:
    \begin{align}
       \budprim = \max_{\bm{x} \in [0,1]^{N}}~ \sum_{n\in [N]} g_{n}\V_{n}x_{n} \text{ s.t. } d\sum_{n\in[N]} g_{n} x_{n} \leq \rho \,,
    \end{align}
    \item \textbf{Step 2.}  We show that $\bm{x}_{\ROI}$ is the unique optimal solution the ``ROI constraint only'' problem:
    \begin{align}
       \roiprim = \max_{\bm{x} \in [0,1]^{N}}~ \sum_{n\in [N]}g_{n}\V_{n} x_{n} \text{ s.t. }\sum_{n\in[N]} g_{n}\left(\V_{n} - \gamma d\right)x_{n} \geq 0\,,
    \end{align}
    \item \textbf{Step 3.}  We show that 
    $\bm{x}_{d} = \min\{\bm{x}_{\bud} ,\bm{x}_{\ROI} \}$ is feasible to $\util(d)$. In other words, we show $\bm{x}_{d}$ is feasible to both $\budprim$ and $\roiprim$. 
\end{itemize}

\paragraph{Step 1.} 
We recognize that $\budprim $ is the  linear program (LP) relaxation of a 0-1 knapsack problem, in which the items' ``value-to-cost ratio'', namely $\frac{g_{n}\V_{n}}{d g_{n}} = \frac{\V_{n}}{d}$ are ordered: $\frac{\V_{1}}{d} > \ldots \frac{\V_{N}}{d}$ since $\V_{1} > \ldots \V_{N} > 0$. Therefore, it is a well known result that the  unique optimal solution to  $\budprim $ is exactly $\bm{x}_{\bud}$ (a threshold vector) defined in the statement of Theorem  \ref{lem:thresholdOptSolext}; see e.g. \cite{dantzig1957discrete} for the optimal solution to the 0-1 knapsack LP relaxation.
 
 \paragraph{Step 2.} 
Let $\widetilde{\bm{x}} \in [0,1]^{N}$ be any optimal solution to \roiprim. Define $\kappa= \max\{n\in [N]:\V_{n}\geq \gamma d\}$ so that $\V_{n}\geq \gamma d$ for all $n \leq \kappa$.  Then it is easy to see for any $n = 1 \ldots \kappa $, we have $\widetilde{x}_{n} = 1$. This is because if there exists some $j \leq \kappa$ such that $\widetilde{x}_{j} < 1$, then the solution $\bm{x} = (\widetilde{x}_{1} \dots \widetilde{x}_{j-1},1,\widetilde{x}_{j+1}, \dots \widetilde{x}_{N})$ is feasible and yields a strictly larger objective than $\widetilde{\bm{x}}$:
\begin{align}
    \sum_{n\in [N]} g_{n}\V_{n}x_{n}  - \sum_{n\in [N]} g_{n}\V_{n}\widetilde{x}_{n} =  \V_{j}(1-\widetilde{x}_{j}) > 0\,.
\end{align}
 Hence, the optimal solution to  $\roiprim$  takes the form of $\widetilde{\bm{x}} = (\underbrace{1 \dots 1}_{\kappa\text{ 1's}}, {y}_{\kappa+1}, \dots {y}_{N}) \in [0,1]^{N}$. Hence, we know that $\widetilde{\bm{y}}\defeq ({y}_{\kappa}, \dots {y}_{N})$ must satisfy
\begin{align}
\label{eq:yptilde}
      \widetilde{\bm{y}} 
       ~\in~ & \arg\max_{\bm{x} \in [0,1]^{N - \kappa}}~ \sum_{k=\kappa+1}^{N}g_{n}\V_{n} x_{n} \text{ s.t. }\sum_{k=\kappa+1}^{N} g_{n}\left(\gamma d - \V_{n} \right)x_{n} \leq \widetilde{{c}}
       \,,
 \end{align}
 where we defined $ \widetilde{{c}} = \sum_{n \in [\kappa]}g_{n}\left(\V_{n} - \gamma d\right) > 0$. Note that we have $\gamma d - \V_{n} > 0$ for all $k = \kappa+1 \dots N$, and hence the optimization problem in Equation \eqref{eq:yptilde} is again an LP relaxation of the 0-1 knapsack problem. Thus similar to Step 1, we again consider the ``value-to-cost-ratios'':  for any  $i,j \in \{\kappa+1 \dots N\}$, we have
 \begin{align*}
    \V_{i} > \V_{j} \Longleftrightarrow \frac{g_{i}\V_{i}}{g_{i}\left(\gamma d -\V_{i}\right)} > \frac{g_{j}\V_{j}}{g_{j}\left(\gamma d -\V_{j}\right)} \,. 
 \end{align*}
 Hence the ``value-to-cost-ratios''$ \frac{g_{n}\V_{n}}{g_{n}\left(\gamma d -\V_{n}\right)}$ decreases in $n$ for $n \in \{\kappa+1 \dots N\} $. Therefore, the optimal solution $\widetilde{\bm{y}}$ to the 0-1 knapsack LP relaxation in Equation \eqref{eq:yptilde} is again unique, and is a threshold vector (again see \cite{dantzig1957discrete}). Hence, the unique optimal solution to \roiprim is a threshold vector, and following Step 1., it is easy to see this unique optimal solution is $\bm{x}_{\ROI}$ defined in the statement of Theorem \ref{lem:thresholdOptSolext}.
 
 \paragraph{Step 3.}
 Since $g_{n}d > 0$ for all $n \in [N]$ and $\bm{x}_{d} = \min\{\bm{x}_{\bud} ,\bm{x}_{\ROI}\}\preceq \bm{x}_{\bud}$, we can apply Lemma \ref{lem:consttrans} (i) with $a_{n} = g_{n}d$, $\bm{Z} =\bm{x}_{\bud} $ and $\bm{Y} = \bm{x}_{d}$, which yields
 \begin{align*}
    d \sum_{n \in [N]} g_{n}x_{d,n} \leq d\sum_{n \in [N]} g_{n} x_{\bud,n} \leq \rho \,,
 \end{align*}
 where the last inequality is due to the fact that $\bm{x}_{\bud}$ is feasible to $\budprim$. This implies $\bm{x}_{d}$ is also feasible to $\budprim$.
 
 On the other hand, again define $\kappa= \max\{n\in [N]:\V_{n}\geq \gamma d\}$ so that $\V_{n}\geq \gamma d$ for all $n \leq \kappa$. Then since $\bm{x}_{d} = \min\{\bm{x}_{\bud},\bm{x}_{\ROI}\}\preceq \bm{x}_{\ROI}$, and 
 since $g_{n}\left(\V_{n}-\gamma d\right) > 0$ for $n = 1\dots \kappa$ and  $g_{n}\left(\V_{n}-\gamma d\right) < 0$ for $n =  \kappa+1 \dots N$, 
  we can apply Lemma \ref{lem:consttrans} (ii) with $b_{n} = g_{n}\left(\V_{n}-\gamma d\right)$, $\bm{Z} =\bm{x}_{\ROI} $ and $\bm{Y} = \bm{x}_{d}$, which shows
 \begin{align*}
      \sum_{n \in [N]} g_{n}\left(\V_{n}-\gamma d\right)x_{\ROI,n} \overset{(i)}{\geq} 0 \overset{(ii)}{\Longrightarrow }\sum_{n \in [N]} g_{n}\left(\V_{n}-\gamma d\right)x_{d,k} \geq 0\,,
 \end{align*}
 where (i) follows from the fact that $\bm{x}_{\ROI}$ is feasible to $\roiprim$ and (ii) follows from the first half of Lemma \ref{lem:consttrans} (ii). Hence $\bm{x}_{d}$ is also feasible to $\roiprim$.

 The rest of the proof is straightforward: $\budprim$, $\roiprim$ and  $\util(d)$ have the same objectives, while each of $\budprim$ and $\roiprim$ has one less constraint than $\util(d)$,
 respectively. So $ \budprim \geq \util(d)$ and $\roiprim \geq \util(d)$. If $\bm{x}_{d} = \bm{x}_{\bud}$, because from Step 3. we know $\bm{x}_{d}$ is feasible to $\util(d)$, then $\budprim = \util(d)$ and $\bm{x}_{d}$
is the optimal solution to both $\budprim$ and $\util(d)$. Similarly, when $\bm{x}_{d} = \bm{x}_{\ROI}$, $\bm{x}_{d}$
is the optimal solution to both $\roiprim$ and $\util(d)$. 

Finally, we argue  $\bm{x}_{d}$ is the unique optimal solution to $\util(d)$. Assume by contradiction there exists some other vector $\bm{x}\in [0,1]^{N}$ that is an optimal solution to $\util(d)$ and $\bm{x}_{d}\neq \bm{x}$. Then, 
again if $\bm{x}_{d} = \bm{x}_{\bud}$, we know that $ \budprim = \util(d)$, and because both $\bm{x}_{d}, \bm{x}$ achieve total value $\util(d)$, then both $\bm{x}_{d}, \bm{x}$ are optimal solutions to $\budprim$, which contradicts uniqueness of the optimal solution to $\budprim$ as argued in Step 1. Similarly, we can again arrive at a contradiction for the case when $\bm{x}_{d} = \bm{x}_{\ROI}$. Hence, the optimal solution to  $\util(d)$ is unique.
 
\halmos

\subsection{Proof of Proposition \ref{prop:reformsellregret}}
\label{pf:prop:reformsellregret}
The proof for this proposition consists of two steps. First, we show that the buyer's optimal hindsight problem w.r.t. a single price $d$, namely $\opt(d)$ in Equation (\ref{lem:thresholdOptSol}) is upper bounded by $T\cdot \util(d)$, which is the single-period myopic optimization problem denoted in Equation (\ref{eq:pricing:buyerutil}). Next, we show  playing the threshold strategy w.r.t. $\bm{x}_{d} \in [0,1]^{N}$ (i.e. the optimal solution to 
$\util(d)$) every period, gives the buyer a total value exactly $T\cdot \util(d)$ while simultaneously satisfying both budget and ROI constraints. Therefore playing the threshold strategy w.r.t. $\bm{x}_{d}$ is the optimal value maximizing strategy to the buyer under a fixed price across all periods.

\paragraph{Step 1.}
Recall the linear program (LP) in Equation (\ref{eq:pricing:buyerutil}) that denotes the buyer's single-period myopic optimization problem. It is easy to see the optimal value is bounded and the LP is feasible (consider the solution with all entries set to be 0). Then, strong duality holds, and therefore for any $d$, there exists corresponding optimal dual variables $(\lambda,\mu) \in \R_{+}^{2}$ s.t. 
\begin{align}
\begin{aligned}
    \util(d) =& \max_{\bm{x}\in[0,1]^{N}}\sum_{n \in [N]}\left(g_{n} (1+\lambda)V_{n} - (\gamma \lambda + \mu) d\right)x_{n} + \rho \mu\\
    =& \sum_{n \in [N]}\left(g_{n} (1+\lambda)V_{n} - (\gamma \lambda + \mu) d\right)_{+} + \rho \mu
    \end{aligned}
\end{align}
On the other hand, when the sequence of posted prices stays constant at $d$, we have
\begin{align}
\label{eq:optbuyersolUB}
\begin{aligned}
    \opt(d) ~\leq~& \max_{\bm{z}\in [0,1]^{T}}\sum_{t\in [T]} \E\left[\left( (1+\lambda)v_{t} - (\gamma \lambda + \mu) d\right)z_{t}\right] + T \rho \mu\\
    ~\leq~& \sum_{t\in [T]}\E\left[ \left( (1+\lambda)v_{t} - (\gamma \lambda + \mu) d\right)_{+}\right] + T \rho \mu\\
    ~=~& T\left(\sum_{n\in[N]} g_{n}\left( (1+\lambda)V_{n} - (\gamma \lambda + \mu) d\right)_{+} + \rho \mu\right)\\
    ~=~& T \cdot \util(d)
    \end{aligned}
\end{align}

\paragraph{Step 2.}
Let $\bm{x}_{d} \in [0,1]^{N}$ be the optimal solution to $\util(d)$ in Equation (\ref{eq:pricing:buyerutil}). Then, the threshold strategy w.r.t $\bm{x}_{d}$ (see Definition \ref{def:threshstrat}) can be represented as 
\begin{align}
\label{eq:threshsol}
    z_{t}^{*}= \sum_{n\in [N]}x_{d,n}\I\{v_{t} = V_{n}\}
\end{align}
 It is easy to see $\{z_{t}^{*}\}_{t\in [T]}$ is feasible to the buyer's optimal hindsight problem $\opt(d)$ because:
 \begin{align}
 \begin{aligned}
 \label{eq:brroigood}
     \E\left[ \sum_{t\in [T]} \left(v_{t} - \gamma d\right)z_{t}^{*}\right] ~=~ &
     \sum_{t\in [T]} \E\left[\left(v_{t} - \gamma d\right)\sum_{n\in [N]}x_{d,n}\I\{v_{t} = V_{n}\}\right]\\
     ~=~ &
     \sum_{t\in [T]} \sum_{n\in [N]}g_{n}(V_{n} - \gamma d)x_{d,n} \overset{(i)}{\geq} 0\\
  \end{aligned}
 \end{align}
 and 
 \begin{align}
 \begin{aligned}
 \label{eq:brbudgetgood}
 \E \left[\sum_{t\in [T]} d z_{t}^{*} \right]~=~ & 
d \sum_{t\in [T]}  E \left[\sum_{n\in [N]}x_{d,n}\I\{v_{t} = V_{n}\}\right]\\
~=~ & 
T\cdot d \sum_{t\in [T]} \sum_{n\in [N]}g_{n}x_{d,n}\\
 \overset{(ii)}{\leq} \rho T
  \end{aligned}
 \end{align}
 where both (i) and (ii) hold because $\bm{x}_{d}$ is feasible to $\util(d)$. Finally, the threshold strategy yields a total value exactly $T\util (d)$ because
 \begin{align}
  \label{eq:brutil}
     \sum_{t\in [T]}\E \left[ v_{t} z_{t}^{*} \right] = \sum_{t\in [T]}
     \E \left[
    v_{t} \sum_{n\in [N]}x_{d,n}\I\{v_{t} = V_{n}\}\right] = T\cdot  \sum_{n\in [N]}g_{n}V_{n}x_{d,n} = T \cdot \util(d)\,,
 \end{align}
 where the final equality follows from the fact that $\bm{x}_{d}$ is optimal to $\util(d)$. 
 
 Therefore, in light of the upper bound shown in Equation (\ref{eq:optbuyersolUB}), the threshold strategy in Equation (\ref{eq:threshsol}) is optimal to the buyer's hindsight problem $\opt(d)$.
 
 Finally, the seller's revenue under the buyer's optimal threshold strategy is $d\sum_{t\in [T]}z_{t}^{*} =  T\cdot  \sum_{n\in [N]}g_{n} x_{d,n} = T\cdot \pi(d)$ where $\pi(d)$ is the per-period revenue defined in Equation (\ref{eq:revenue}).
 \halmos

 \subsection{Proof of Theorem \ref{thm:pricing:optd}}
\label{pf:thm:pricing:optd}

Our proof relies on the following fact
\begin{fact}
\label{fact:nonbinding}
If price $d$ is nonbinding, then the corresponding optimal solution $\bm{x}_{d}$ to $\util(d)$ is $\bm{x}_{d} = (1 \dots 1) \in \R_{+}^{n}$.
\end{fact}
\begin{proof}
We prove the claim via contradiction.
Assume there is some index $k\in [N]$ such that $\bm{x}_{d,k} < 1$. Then consider the solution $\bm{x} = (x_{d,1}\dots x_{d,k-1}, y ,x_{d,k+1},\dots x_{d,n})$ where we replaced the $k$'th entry of $\bm{x}_{d}$ with 
\begin{align}
    y = x_{d,k}+
    \epsilon,\quad \text{ where } \epsilon \defeq \min\left\{1 - x_{d,k}, \frac{\rho - \sum_{n\in [N]}\valpe_{n}x_{d,n}}{d\valpe_{k}}, \frac{\sum_{n\in [N]}\left(\V_{n}- \gamma d\right)\valpe_{n}x_{d,n}}{\left|\V_{k}- \gamma d\right|\valpe_{k}} \right\} \overset{(i)}{>} 0\,, \notag
\end{align}
where (i) follows from the fact that $\bm{x}_{d}$ is nonbinding, i.e. $\rho > \sum_{n\in [N]}\valpe_{n}x_{d,n}$
and $\sum_{n\in [N]}\left(\V_{n}- \gamma d\right)\valpe_{n}x_{d,n}> 0$. Then 
\begin{align}
    d\sum_{n\in [N]}\valpe_{n}x_{n} ~=~ &  
    d\sum_{n\in [N]}\valpe_{n}x_{d,n} + d\valpe_{k}\epsilon
    ~\leq~  d\sum_{n\in [N]}\valpe_{n}x_{d,n} + \left(\rho - \sum_{n\in [N]}\valpe_{n}x_{d,n}\right) = \rho\,. \notag
\end{align}
On the other hand, if $\V_{k}- \gamma d > 0$, then 
\begin{align}
  \sum_{n\in [N]}\left(\V_{n}- \gamma d\right)\valpe_{n}x_{d,n}= \sum_{n\in [N]}\left(\V_{n}- \gamma d\right)\valpe_{n}x_{d,n} + \left(\V_{k}- \gamma d\right)\valpe_{k} \epsilon  > \sum_{n\in [N]}\left(\V_{n}- \gamma d\right)\valpe_{n}x_{d,n} > 0\,.\notag
\end{align}
If $\V_{k}- \gamma d < 0$, then 
\begin{align*}
  \sum_{n\in [N]}\left(\V_{n}- \gamma d\right)\valpe_{n}x_{d,n}~=~ & \sum_{n\in [N]}\left(\V_{n}- \gamma d\right)\valpe_{n}x_{d,n} + \left(\V_{k}- \gamma d\right)\valpe_{k} \epsilon\\ 
  ~\geq~ & \sum_{n\in [N]}\left(\V_{n}- \gamma d\right)\valpe_{n}x_{d,n} + \left(\V_{k}- \gamma d\right)\cdot \frac{\sum_{n\in [N]}\left(\V_{n}- \gamma d\right)\valpe_{n}x_{d,n}}{\left|\V_{k}- \gamma d\right|} = 0
\end{align*}
where in the last equality we used $\left|\V_{n}- \gamma d\right| = - \left(\V_{n}- \gamma d\right)$ since $\V_{n}- \gamma d < 0$. 

The above shows $\bm{x}$ is feasible to $\util(d)$. On the other hand, $\sum_{n \in [N]}\V_{n}\valpe_{n}x_{d,n} < \sum_{n \in [N]}\V_{n}\valpe_{n}x_{n} $, so $\bm{x}$ yields a strictly larger objective than $\bm{x}_{d}$, contradicting the optimality of $\bm{x}_{d}$.
\end{proof}

We now return to our proof for Theorem \ref{thm:pricing:optd}.

\textit{(1).}
When both $d,\widetilde{d}$ are non-binding, Fact \ref{fact:nonbinding} implies $\bm{x}_{d} =\bm{x}_{\widetilde{d}}= (1 \dots 1)$. 
\begin{align*}
   \rev(d) = d  \sum_{n\in[N]}\valpe_{n}x_{d,n} = d\sum_{n\in[N]}\valpe_{n} < \widetilde{d}\sum_{n\in[N]}\valpe_{n} = \widetilde{d}\sum_{n\in[N]}\valpe_{n}x_{\widetilde{d},n}  = \rev(\widetilde{d})\,.
\end{align*}

\textit{(2).}
We prove this claim by contradiction. Assume $\widetilde{d}$ is non-binding and $\widetilde{d} > d$ where $d$ is budget binding. Fact \ref{fact:nonbinding} states that  $\bm{x}_{\widetilde{d}}= (1 \dots 1)$.
Hence 
\begin{align}
    \rho = \rev(d) = d  \sum_{n\in[N]}\valpe_{n}x_{d,n} \leq d\sum_{n\in[N]}\valpe_{n}x_{\widetilde{d},n}  
    < \widetilde{d}\sum_{n\in[N]}\valpe_{n}x_{\widetilde{d},n}  \overset{(i)}{<} \rho\,, \notag
\end{align}
where (i) follows from the definition that  $\widetilde{d}$  is non-binding. Hence we obtain a contradiction, and
$\widetilde{d}$ cannot be non-binding. This means $\widetilde{d}$ must be budget or ROI binding. 

\textit{(3).}
Here we show that if some price $d \in \mathcal{D}$ is ROI binding so that $\sum_{n \in [N]}(\V_{n} - \gamma d)\valpe_{n}x_{d,n} = 0$, any price
$\widetilde{d} > d$ must also be ROI binding. We first claim that $\bm{x}_{\widetilde{d}}\preceq \bm{x}_{d}$. To show this, we use a contradiction argument by assuming $\bm{x}_{\widetilde{d}}\succeq \bm{x}_{d}$.

Let the threshold vector $\bm{x}_{d} $ be characterized by $\bm{x}_{d} = \thr(J,q)$ (see definition of threshold vectors in Definition \ref{def:thresholdvect}). Under Assumption \ref{assum:price:nontriv}, we note that $\bm{x}_{d}$ cannot have all 0 entries and hence $x_{d,1} > 0$. However, since $\sum_{n \in [N]}(\V_{n} - \gamma d)\valpe_{n}x_{d,n} = 0$, it must be the case that $\V_{J+1} - \gamma d < 0$. Now, applying the ordering property for threshold vectors in the second half of Lemma \ref{lem:consttrans} (ii) by taking $\bm{Z} =\bm{x}_{\widetilde{d}} $,  $\bm{Y} = \bm{x}_{d}$, and $b_{i} = \V_{i} - \gamma d$ we have
 \begin{align}
     0= \sum_{n \in [N]}(\V_{n} - \gamma d)\valpe_{n}x_{d,n} \geq \sum_{n \in [N]}(\V_{n} - \gamma d)\valpe_{n}x_{\widetilde{d},n}> \sum_{n \in [N]}(\V_{n} - \gamma \widetilde{d})\valpe_{n}x_{\widetilde{d},n} \,. \notag
 \end{align}
 In the last inequality we used 
the fact that $\widetilde{d} > d$. Hence, this contradicts the feasibility of 
$\bm{x}_{\widetilde{d}}$, so we conclude that $\bm{x}_{\widetilde{d}}\preceq \bm{x}_{d}$. This further implies
\begin{align}
   \rho \geq  \underbrace{d\sum_{n \in [N]}\valpe_{n}x_{d,n}}_{=\rev(d)} \overset{(i)}{=} \frac{1}{\gamma}\sum_{n \in [N]}\V_{n} \valpe_{n}x_{d,n}  \overset{(ii)}{>} \frac{1}{\gamma}\sum_{n \in [N]}\V_{n} \valpe_{n}x_{\widetilde{d},n} \overset{(iii)}{\geq} \underbrace{\widetilde{d}\sum_{n \in [N]}\valpe_{n}x_{\widetilde{d},n}}_{=\rev(\widetilde{d})}\,, \notag 
\end{align}
where (i) follows from $d$ being 
ROI binding, i.e. $\sum_{n \in [N]}(\V_{n} - \gamma d)\valpe_{n}x_{d,n} = 0$;  (ii) follows from $\bm{x}_{\widetilde{d}}\preceq \bm{x}_{d}$; (iii) follows from feasibility of $\widetilde{d}$ so that $\sum_{n \in [N]}(\V_{n} - \gamma \widetilde{d})\valpe_{n}x_{\widetilde{d},n} \geq 0$. Therefore, $\rho \geq \rev(d) > \rev(\widetilde{d})$.

Finally, $\rho > \rev(\widetilde{d})$ implies that 
$\widetilde{d}$  is either non-binding or ROI binding. We note that it is not possible for $\widetilde{d}$ to be non-binding, because 
$\widetilde{d}$ non-binding implies $\bm{x}_{\widetilde{d}}= (1 \dots 1)$ according Fact \ref{fact:nonbinding}, contradicting $\bm{x}_{\widetilde{d}}\preceq \bm{x}_{d}$ which we showed earlier. Here we used the fact that 
$\bm{x}_{d} \neq (1 \dots 1)$ because $\bm{x}_{d}$ is ROI binding and Assumption  \ref{assum:price:nontriv} states for any $d \in \mathcal{D}$, 
$\sum_{n \in [N]}(\V_{n} - \gamma d)\valpe_{n} \neq 0$.
\halmos

\subsection{Additional lemmas for Section \ref{sec:reform}}
\begin{lemma}[Ordering property for threshold vectors]
 \label{lem:consttrans}
 Consider $\{a_{i}\}_{i \in [N]} \subseteq \R_{+}^{N}$ and  $\{b_{i}\}_{i \in [N]} \subseteq \R^{N}$ where there exists some $j \in [N]$ such that $b_{i} > 0$ for all $i =1\dots j$ and $b_{i}< 0$ for all $i = j+1 ,\dots m$. Let $\bm{Z},\bm{Y} \in [0,1]^{N}$ be two threshold vectors (see Definition \ref{def:thresholdvect}) such that $\bm{Y} = \thr(\J_{Y},q_{Y})$, $\bm{Z} = \thr(\J_{Z},q_{Z})$, and $\bm{Z} \succeq\bm{Y}$. Then the following hold: 
 \begin{itemize}
     \item[(i)] $\sum_{i \in [N]} a_{i}Z_{i}\geq \sum_{i \in [N]} a_{i}Y_{i}$. 
     \item[(ii)] If 
     $\sum_{i \in [N]} b_{i}Z_{i}\geq 0$ then $\sum_{i \in [N]} b_{i}Y_{i} \geq 0$. Furthermore, if $b_{\J_{Y}+1} < 0$, then $\sum_{i \in [N]} b_{i}Y_{i} \geq \sum_{i \in [N]} b_{i}Z_{i} \geq 0$.
     \item[(iii)] If
     $\sum_{i \in [N]} b_{i}Y_{i} < 0$ then $\sum_{i \in [N]} b_{i}Z_{i} < 0$.
 \end{itemize}
 \end{lemma}
 
 \textit{Proof.}
 \paragraph{(i)}Since $a_{i} > 0$ for all $i \in [N]$, and $\bm{Z}\succeq \bm{Y}$ (i.e. $Z_{i}\geq Y_{i}$ for all $i\in [N]$), it is easy to see $\sum_{i\in [N]}a_{i}Z_{i} \geq \sum_{i\in [N]}a_{i}Y_{i}$.

\paragraph{(ii)}

By the definition of threshold vectors, we have $Y_{J_{Y}+1} = q_{Y}$ while $Y_{i} = 0$ for all $i > J_{Y}+1$. We prove the claim by contradiction by 
assuming $\sum_{i\in [N]}b_{i}Y_{i} <0 $.

First,  it is easy to see $b_{J_{Y}+1} <0$. This is because if $b_{J_{Y}+1} >0$, then $b_{i} >0$ for all $i = 1 \dots J_{Y}+1$ by the definition of $\{b_{i}\}_{i\in[N]}$, and hence $\sum_{i\in [N]}b_{i}Y_{i} = \sum_{i\in [J_{Y}+1]}b_{i}Y_{i} \geq 0$ contradicting our assumption that $\sum_{i\in [N]}b_{i}Y_{i} <0 $ . 

Next, since $\sum_{i\in [N]}b_{i}Y_{i}  < 0\leq \sum_{i\in [N]}b_{i}Z_{i}$, we have $\sum_{i\in [N]}b_{i}(Z_{i} - Y_{i}) \geq 0$. On the other hand,

\begin{align*}
    \sum_{i\in [N]}b_{i}(Z_{i} - Y_{i}) \overset{(i)}{=} \sum_{i=J_{Y}+1}^{N}b_{i}(Z_{i} - Y_{i}) \overset{(ii)}{<} 0\,.
\end{align*}
Here, (i) follows from the definition of a threshold vector so that $Y_{i} = 1$ for all $i = 1\dots J_{Y}$ and also $Z_{i} =1$ for all $i = 1\dots J_{Y}$ due to $\bm{Z}\succeq \bm{Y}$. (ii) follows from the fact that 
$b_{J_{Y}+1} < 0$ so $b_{i} < 0$ for all $i \geq J_{Y}+1$ due to the definition of  $\{b_{i}\}_{i\in[N]}$. Hence, we arrive at a contradiction, which allows us to conclude
{the first half of the claim, i.e.
     $\sum_{i \in [N]} b_{i}Z_{i}\geq 0$ implies $\sum_{i \in [N]} b_{i}Y_{i} \geq 0$.}

{We now show the second half of the claim i.e. $b_{J_{Y}+1} < 0$ implies $\sum_{i \in [N]} b_{i}Y_{i} \geq \sum_{i \in [N]} b_{i}Z_{i} \geq 0$.} If $b_{J_{Y}+1}< 0$, then $b_{i} < 0$ for all $i =J_{Y}+1 + \dots J_{Z}+1 $, and hence
\begin{align*}
    \sum_{i\in [N]}b_{i}(Z_{i} - Y_{i}) =b_{J_{Y}+1}(Z_{J_{Y}+1} - Y_{J_{Y}+1}) +  \sum_{i=J_{Y}+2}^{J_{Z}+1}b_{i}Z_{i} \overset{(i)}{<} 0\,.
\end{align*}
Note that in the above inequality the summand $ \sum_{i=J_{Y}+2}^{J_{Z}+1}b_{i}Z_{i} $ does not exist if $J_{Y} = J_{Z}$, and in (i) we also used the fact that $Y_{i} = 0$ for all $i > J_{Y}+1$ using the definition of a threshold veector.

\paragraph{(iii)} 
We again use a contradiction argument by assuming $\sum_{i\in [N]}b_{i}Z_{i} \geq 0$, and the rest of the proof is almost identical to that of (ii) so we will omit it here.
\halmos

\section{Proofs for Section \ref{sec:pricingalgo}}
\label{app:pricingalgo}

\subsection{Proof of Theorem \ref{thm:episodicprice}}
\label{proof:thm:episodicprice}
Define $ G := \min_{d,\widetilde{d} \in \mathcal{D}:\rev(d) \neq \rev(\widetilde{d})}\left|\rev(d) - \rev(\widetilde{d})\right|$ to be the minimum revenue gap for all price pairs that do not yield the same revenue, where  $\rev(d) \defeq d\sum_{n\in [N]} \valpe_{n} x_{d,n}$ for any $d\in \mathcal{D}$ is the per-period average seller revenue  defined in Equation (\ref{eq:revenue}). Recall $\Hat{\rev}(D_{h}) = \frac{D_{h}}{|\mathcal{E}_{h}|}\sum_{t\in \mathcal{E}_{h}} z_{t}$ is the estimate of  $\rev(D_{h})$
for  episode $h $ with fixed price $D_{h}$ (see Algorithm \ref{alg:pricing}). 

For any exploration episode $\mathcal{E}_{h}$ whose length is $|\mathcal{E}_{h}| = T^{1-\xi + \epsilon}$, we have w.p. at least $1-1/T$
\begin{align}
\begin{aligned}
\label{eq:adapt:boundedness}
   &  \left|\frac{\Hat{\rev}(D_{h})}{D_{h}}- \frac{\rev(D_{h})}{D_{h}}\right| = \left|\frac{1}{|\mathcal{E}_{h}|}\sum_{t\in \mathcal{E}_{h}}z_{t} - \frac{\rev(D_{h})}{D_{h}}\right|\overset{(i)}{\leq} \frac{\phi(|\mathcal{E}_{h}|)}{|\mathcal{E}_{h}|} \overset{(ii)}{\leq} \frac{\phi(T)}{T^{1-\xi + \epsilon}}\\
   \overset{(iii)}{\Rightarrow} & \left|\Hat{\rev}(D_{h})-\rev(D_{h})\right|\leq \frac{\phi(T)}{T^{1-\xi + \epsilon}}
   \end{aligned}
\end{align}
where (i) is due to the definition of $\xi$-adaptive buyers in Definition \ref{def:adapt}; (ii) is due to the fact that $\phi$ is an increasing function and the exploration episode lengths are $|\mathcal{E}_{h}| = T^{1-\xi + \epsilon}$; (iii) is due to the fact that all prices are less than 1.

Since $\phi(T) = \mathcal{O}(T^{1-\xi})$, there exists some $T_{\epsilon} \in \N$ such that when $T> T_{\epsilon}$ we have
\begin{align}
\label{eq:buyererrsmallenough}
    \frac{\phi(T)}{T^{1-\xi + \epsilon}} < \frac{G}{2}
\end{align}


The rest of the proof relies on the following lemma:
\begin{lemma}
\label{fact:takeproberror}
Assume $T > T_{\epsilon}$ s.t. Equation \eqref{eq:buyererrsmallenough} holds. If $ \Hat{\rev}(D_{i}) \geq \Hat{\rev}(D_{j})$ for some exploration episodes $i,j$ s.t. $i\neq j$, then w.p. at least $1-\frac{1}{T}$, $\rev(D_{i}) \geq \rev(D_{j})$. Furthermore, the following event $\mathcal{G}$
\begin{align}
\label{eq:goodbinsearch}
    \mathcal{G} =\left\{\Hat{\rev}(D_{i}) \geq \Hat{\rev}(D_{j}) \Longrightarrow \rev(D_{i}) \geq \rev(D_{j}) ~\text{ for all exploration episodes } i\neq j\right\}
\end{align}
 holds with probability at least $1-\frac{H(H-1)}{2T}$, where $H= \left\lfloor\log_{2}(M)\right\rfloor + 1$ is the maximum number of binary search iterations (i.e. number of episodes in the exploration phase).
\end{lemma}
\begin{proof}[Proof of Lemma \ref{fact:takeproberror}]
Because $\Hat{\rev}(D_{i}) \geq \Hat{\rev}(D_{j})$, applying Equation (\ref{eq:adapt:boundedness}) for episodes $i,j$ yields
\begin{align}
   \rev(D_{i}) +\frac{\phi(T)}{T^{1-\xi + \epsilon}}  \geq \Hat{\rev}(D_{i}) \geq \Hat{\rev}(D_{j}) \geq \rev(D_{j}) - \frac{\phi(T)}{T^{1-\xi + \epsilon}} \Longrightarrow  \frac{2\phi(T)}{T^{1-\xi + \epsilon}} \geq \rev(D_{j}) - \rev(D_{i}) \,, \notag 
\end{align}

Now, contrary to our claim, suppose that $ \rev(D_{i}) <\rev(D_{j}) $. We then have  
\begin{align*}
   \frac{2\phi(T)}{T^{1-\xi + \epsilon}} \geq \rev(D_{j}) - \rev(D_{i}) \geq \min_{d,\widetilde{d} \in \mathcal{D}:\rev(d) \neq \rev(\widetilde{d})}\left|\rev(d) - \rev(\widetilde{d})\right|:= G\,, 
\end{align*}
which contradicts Equation \eqref{eq:buyererrsmallenough} for $T > T_{\epsilon}$. As there are $H(H-1)/2$ pairs $(i,j)$ such that $i\neq j$, a simple union bound shows event $\mathcal{G}$ holds with probability at least $1-\frac{H(H-1)}{2T}$.
\end{proof}

We now return to our proof of Theorem \ref{thm:episodicprice}. We first show that under event $\mathcal{G}$ (see Equation \eqref{eq:goodbinsearch}),the final price in the exploitation phase $\D_{m^{*}}$ is revenue-optimal, i.e. $\max_{d\in \mathcal{D}}\rev(d) =\rev\left(\D_{\term}\right)$


We use an induction argument that shows after each iteration of the binary search procedure in the exploration phase of Algorithm \ref{alg:pricing}, $\rev(\D_{m}) \leq \rev(\D_{\term})$ for all $m \leq \lef$ and $m\geq \rig$.  The base case is the first iteration, where we have
$\lef = 1$, $\rig = M $. If $\term = \lef = 1$, then  under event $\mathcal{G}$ we get
\begin{align*}
    \Hat{\rev}(\D_{1}) \geq \Hat{\rev}(\D_{M}) \overset{(i)}{\implies} \rev(\D_{1}) \geq \rev(\D_{M}) \,.
\end{align*}
 Hence after the first iteration $\rev(\D_{m}) \leq \rev(\D_{\term})$ for any
$m \leq \lef$ and $m\geq \rig$.  The case for $\term = \rig$ follows from the same argument.

Now assume that the induction hypothesis holds, i.e. at the beginning of some iteration  with the  tuple $(\lef,\rig, \term)$, we have $\rev(\D_{m}) \leq \rev(\D_{\term})$
$m \leq \lef$ and $m\geq \rig$. According to Algorithm \ref{alg:pricing}, we only need to show two cases in order to validate the induction procedure.
\begin{itemize}
    \item \textbf{Case 1.} If $ \Hat{\rev}(\D_{\med}) < \Hat{\rev}(\D_{\med+1})$, then we show $  \rev(\D_{m}) \leq \rev(\D_{\med+1}) $ for all $m = 1 \dots \med +1$
     \item \textbf{Case 2.} If $ \Hat{\rev}(\D_{\med}) \geq \Hat{\rev}(\D_{\med+1})$, then we show $\rev(\D_{m}) \geq \rev(\D_{\med}) $ for all $m = \med +1 \dots M$
\end{itemize}
Note that under Case 1., $\med+1$ will be the new value of $m^{*}$ in the next iteration (i.e. the next induction step). So by showing $ \rev(\D_{m}) \leq \rev(\D_{\med+1}) $ for all $m = 1 \dots \med +1$, we validate the induction hypothesis for the next induction step. A similar argument holds for Case 2.

\textbf{Case 1.} When $\Hat{\rev}(\D_{\med}) < \Hat{\rev}(\D_{\med+1})$, under event $\mathcal{G}$ (see Equation \eqref{eq:goodbinsearch}) we have $ \rev(\D_{\med}) \leq \rev(\D_{\med+1})$. We claim that $\D_{\med}$ cannot be an ROI binding price. Assume the contrary that $\D_{\med}$ is ROI binding. Then, part (3) of Theorem \ref{thm:pricing:optd}  states $\rev(\D_{\med+1}) < \rev(\D_{\med})$, leading to a contradiction. Hence $\D_{\med}$ must be either a nonbinding price or a budget binding price. Applying part (1) of Theorem \ref{thm:pricing:optd}, we can then conclude that for any $m \leq \med$, $\rev(\D_{m}) \leq \rev(\D_{\med})$, so
\begin{align}
    \rev(\D_{m}) \leq \rev(\D_{\med})\leq \rev(\D_{\med+1}) \quad \forall m = 1 \dots \med\,. \notag 
\end{align}
At the end of the iteration, as we update ${\term}^{+}= \med+1$ (here we denote ${\term}^{+}$ as the updated value to distinguish from its initial value at the start of the iteration), we have $\rev(\D_{{\term}^{+}}) \geq  \rev(\D_{\med+1}) \geq  \rev(\D_{\med}) \dots \rev(\D_{1})$. On the other hand, since $\Hat{\rev}(\D_{{\term}^{+}}) = \max_{m \in \{\term,\med+1\}} \Hat{\rev}(\D_{m}) \geq \Hat{\rev}(\D_{\term})$, event $\mathcal{G}$ implies
\begin{align}
    \rev(\D_{{\term}^{+}}) \geq \rev(\D_{\term}) \overset{(i)}{\geq} \rev(\D_{m})  \quad \forall m = \rig \dots M\,,  \notag 
\end{align}
 where (i) follows from the induction hypothesis. Therefore, we have
 \begin{align}
     \rev(\D_{{\term}^{+}}) \geq \rev(\D_{m}) \quad \forall m = \rig \dots M \text{ and } m =  1 \dots \med+1\,,\notag 
 \end{align}
 and by realizing the tuple $( \med+1, \rig, {\term}^{+})$ is the initial tuple for the next iteration concludes the induction step. 
 
\textbf{Case 2.} The case when $\Hat{\rev}(\D_{\med}) \geq \Hat{\rev}(\D_{\med+1})$ follows from an identical argument, and we will omit the details. This concludes the  induction proof.

The above implies that when the event $\mathcal{G} =\left\{\Hat{\rev}(D_{i}) \geq \Hat{\rev}(D_{j}) \Longrightarrow \rev(D_{i}) \geq \rev(D_{j}) ~\text{ for all }i,j \in [H]\right\}$ holds throughout the exploration phase, the above induction argument implies we have $ \rev(\D_{\term}) \geq  \rev(\D_{m})$for all $m \in [M]$. Hence $\rev(\D_{\term}) = \max_{d\in \mathcal{D}}\rev(d)$ w.p. at least $1- \frac{H(H-1)}{2T}$ according to Lemma \ref{fact:takeproberror} where $H= \left\lfloor\log_{2}(M)\right\rfloor + 1$.

Furthermore, we point out that in each iteration of the binary search procedure the seller explores at most two prices. Hence the entire 
exploration phase, which consists of all periods in exploration episodes and we denote as $\mathcal{E}$, has length at most $2E \left(\left\lfloor\log_{2}(M)\right\rfloor + 1\right) = 2T^{1-\xi + \epsilon}\left(\left\lfloor\log_{2}(M)\right\rfloor + 1\right)$ periods. Therefore, the seller's regret can be upper bounded as 
\begin{align*}
  \text{Reg}_{\text{sell}} 
    ~=~ & T\max_{d\in \mathcal{D}}\rev(d) - \sum_{t\in[T]}\E\left[d_{t}z_{t}\right] \\
    ~\leq~ & |\mathcal{E}| + \sum_{t =|\mathcal{E}| + 1}^{T}\max_{d\in \mathcal{D}}\rev(d) - \E\left[d_{t}z_{t}\right] \\ 
    ~\overset{(i)}{\leq}~ &  |\mathcal{E}| + 
     \sum_{t\in[T]/\mathcal{E} }\E\left[\left(\rev(\D_{\term}) - \D_{\term}z_{t}\right) \I\left\{\mathcal{G} \right\}\right]  + (T- |\mathcal{E}|) \prob\left(\mathcal{G}^{c}\right)\\ 
     ~\leq~ &  |\mathcal{E}| + 
    \D_{\term}(T- |\mathcal{E}|)\cdot \E\left[\frac{\rev(\D_{\term})}{\D_{\term}} - \frac{1}{T - |\mathcal{E}|}\sum_{t\in[T]/\mathcal{E}}z_{t} \right]  + (T- |\mathcal{E}|) \prob\left(\mathcal{G}^{c}\right)\\ 
      ~\overset{(ii)}{\leq}~  &  |\mathcal{E}| +
    \phi\left(T-|\mathcal{E}|\right) + (T-|\mathcal{E}|) \cdot \prob\left(\left|\frac{\rev(\D_{\term})}{\D_{\term}} - \frac{1}{T - |\mathcal{E}|}\sum_{t\in[T]/\mathcal{E}}z_{t} > \frac{\phi(T-|\mathcal{E}|)}{T-|\mathcal{E}|}\right|\right)+ T \prob\left(\mathcal{G}^{c}\right)\\ 
     ~\overset{(iii)}{\leq}~  &  |\mathcal{E}| +
    \phi\left(T-|\mathcal{E}|\right) + 1 + T \prob\left(\mathcal{G}^{c}\right)\\ 
     ~\overset{(iv)}{\leq}~ & 2 \left(\left\lfloor\log_{2}(M)\right\rfloor + 1\right) \cdot T^{1-\xi + \epsilon}  + 
    \phi\left(T\right) + \left(\left\lfloor\log_{2}(M)\right\rfloor + 1\right)^{2}/2\,.
\end{align*}
In (i) we used the fact that $\max_{d\in \mathcal{D}}\rev(d)= \rev(\D_{\term})$ under event $\mathcal{G}$ and $d_{t} = \D_{\term}$ for all exploitation periods $t \in [T]/\mathcal{E}$; in (ii) and (iii) we used the definition of $\xi$-adaptive buyer algorithm (see Definition \ref{def:adapt})  so that for the exploitation phase $[T]/\mathcal{E}$, the event
$\left|\frac{\rev(\D_{\term})}{\D_{\term}} - \frac{1}{T - |\mathcal{E}|}\sum_{t\in[T]/\mathcal{E}}z_{t}\right| \leq \frac{\phi(T-|\mathcal{E}|)}{T-|\mathcal{E}|}$ holds with probability at least $1-1/T$, and also $\phi$ is an increasing function;
In (iv), we used the fact that all periods in exploration episodes $\mathcal{E}$, has length at most $2E \left(\left\lfloor\log_{2}(M)\right\rfloor + 1\right) = 2T^{1-\xi + \epsilon}\left(\left\lfloor\log_{2}(M)\right\rfloor + 1\right)$ periods, and the fact that 
$\prob\left(\mathcal{G}^{c}\right) \leq \frac{\left(\left\lfloor\log_{2}(M)\right\rfloor + 1\right)\cdot \left\lfloor\log_{2}(M)\right\rfloor }{2T}$ according to Lemma \ref{fact:takeproberror}, so $1 + T \prob\left(\mathcal{G}^{c}\right)\leq \left(\left\lfloor\log_{2}(M)\right\rfloor + 1\right)^{2}/2$ given $M \geq 2$.
\qed

\section{Proofs for Section \ref{sec:buyeralgo}}
\label{app:buyeralgo}
\subsection{Proof of Lemma \ref{lem:bestRespondAdapt}}
\label{pf:lem:bestRespondAdapt}

Recall that when the buyer best responds, she adopts the threshold strategy w.r.t $\bm{x}_{d}$ where $\bm{x}_{d} \in [0,1]^{N}$ is the optimal solution to $\util(d)$ in Equation \eqref{eq:pricing:buyerutil}; see Definition \ref{def:br} for best response. Further, the threshold strategy can be represented as decision 
\begin{align*}
    z_{t}^{*} = \sum_{n\in [N]}x_{d,n}\I\{v_{t} = V_{n}\}\,.
\end{align*}
Then, for any exploration or exploitation episode $\mathcal{E}$ (whose posted price we denote as $d$), for the best response decisions $\{z_{t}\}_{t\in \mathcal{E}}$ defined above, we have for any $t\in \mathcal{E}$
\begin{align*}
    \E[d z_{t}^{*}] = d \sum_{n\in [N]}g_{n}x_{d,n} = \rev(d)
\end{align*}
where $\rev(d)$ is the per-period expected revenue defined in Equation \eqref{eq:revenue}. Hence, by defining
\begin{align*}
    Y_{t} = d z_{t}^{*} - \rev(d)
\end{align*}
we know that the sequence $\{Y_{t}\}_{t\in \mathcal{E}}$ is a martingale difference sequence such that that $|Y_{t}|\leq d\leq 1$ for all $t\in \mathcal{E}$. By Azuma Hoeffding's inequality (see Lemma \ref{lem:azuma}) we have for any $\delta \in (0,1)$
\begin{align*}
    \prob\left(\left|d\sum_{t\in \mathcal{E}} z_{t}^{*} - |\mathcal{E}|\cdot \rev(d)\right| > \sqrt{2|\mathcal{E}|\log\left(2/\delta\right)}\right)\leq \delta \,.
\end{align*}
Hence, by taking $\delta = 1/T$ and considering the increasing function $\phi(x) =\sqrt{2x\log\left(2T\right)}= \mathcal{O}(x^{1/2})$, for any exploration/exploitation episode $\mathcal{E}$ (whose price we denote as $d$) we have 
\begin{align*}
    \left| \frac{d}{|\mathcal{E}|} \sum_{t\in \mathcal{E}}z_{t}^{*} - \rev(d)\right| \leq \frac{\phi (|\mathcal{E}|)}{|\mathcal{E}|}
\end{align*}
with probability (w.p.) at least $1-1/T$. Therefore best responding is $\frac{1}{2}$-adaptive. 
\qed

\subsection{Proof of Theorem \ref{corr:bestRespondBuyer}}
\label{pf:corr:bestRespondBuyer}
For the seller, the regret upper bound is a direct result of Lemma \ref{lem:bestRespondAdapt} and Theorem \ref{thm:episodicprice}. 

On the other hand for the buyer, following the exact proof Step 2. in the proof of Proposition \ref{prop:reformsellregret} (see Appendix 
\ref{pf:prop:reformsellregret}), in particular Equations \eqref{eq:threshsol}, \eqref{eq:brroigood} and  \eqref{eq:brbudgetgood}, we know that by best responding, the buyer's budget and ROI constraints are satisfied. Finally, we bound the buyer's regret (see Definition in \eqref{eq:buyerregret}) as followed:

Let $d$ be the posted price in the final exploitation episode (see Algorithm \ref{alg:pricing}). Using an argument similar to Step 1. in the proof of Proposition \ref{prop:reformsellregret}, for the linear program (LP) $\util(d)$ in Equation (\ref{eq:pricing:buyerutil}) that denotes the buyer's single-period myopic optimization problem, it is easy to see the optimal value is bounded and the LP is feasible (consider the solution with all entries set to be 0). Then, strong duality holds, and there exists corresponding optimal dual variables $(\lambda,\mu) \in \R_{+}^{2}$ w.r.t. the exploitation price $d$ s.t. 
\begin{align}
\label{eq:brsingleutil}
\begin{aligned}
    \util(d) =& \max_{\bm{x}\in[0,1]^{N}}\sum_{n \in [N]}g_{n}\left( (1+\lambda)V_{n} - (\gamma \lambda + \mu) d\right)x_{n} + \rho \mu\\
    =& \sum_{n \in [N]}g_{n}\left( (1+\lambda)V_{n} - (\gamma \lambda + \mu) d\right)_{+} + \rho \mu
    \end{aligned}
\end{align}

Similar to Equation \eqref{eq:optbuyersolUB}, by denoting $\mathcal{E}$ to be all periods within exploration episodes, the buyer's hindsight objective can be bounded as
\begin{align}
\begin{aligned}
\label{eq:brbenchmarkUB}
 \opt(\bm{d}_{1:T}) ~\leq~& \max_{\bm{z}\in [0,1]^{T}}\sum_{t\in [T]} \E\left[\left( (1+\lambda)v_{t} - (\gamma \lambda + \mu) d\right)z_{t}\right] + T \rho \mu\\
    ~\leq~& \sum_{t\in [T]}\E\left[ \left( (1+\lambda)v_{t} - (\gamma \lambda + \mu) d\right)_{+}\right] + T \rho \mu\\
    ~=~& \sum_{t\in[T]}\left(\sum_{n\in[N]} g_{n}\left( (1+\lambda)V_{n} - (\gamma \lambda + \mu) d\right)_{+} + \rho \mu\right)\\
    ~\leq ~& (1+\lambda + \rho \mu) \cdot |\mathcal{E}| + \sum_{t\in[T]/\mathcal{E}}\left(\sum_{n\in[N]} g_{n}\left( (1+\lambda)V_{n} - (\gamma \lambda + \mu) d\right)_{+} + \rho \mu\right)\\
     ~\overset{(i)}{=} ~& \Theta(T^{\frac{1}{2}+\epsilon}) + (T-|\mathcal{E}|)\util(d)\,.
     \end{aligned}
\end{align}
Here (i) follows from Equation \eqref{eq:brsingleutil} and the fact that there are at most $2\left(\left\lfloor\log_{2}(M)\right\rfloor + 1\right)$ exploration episodes, which implies in $\mathcal{E}$
there are at most $2T^{1-\xi + \epsilon}\left(\left\lfloor\log_{2}(M)\right\rfloor + 1\right) = \Theta(T^{\frac{1}{2}+\epsilon})$ periods. The buyer's regret can be thus bounded as followed
\begin{align*}
     \buyreg = \opt(\bm{d}_{1:T}) - \sum_{t\in [T]}\E\left[v_{t}z_{t}\right] \leq \Theta(T^{\frac{1}{2}+\epsilon}) + (T-|\mathcal{E}|)\util(d) - \sum_{t\in [T]/\mathcal{E}}\E\left[v_{t}z_{t}\right]
    = \Theta(T^{\frac{1}{2}+\epsilon})
\end{align*}
where in the final equality, we used the fact that the buyer's expected utility is exactly $\util(d)$ for each exploitation period when best responding as shown in  Equation \eqref{eq:brutil}.
\qed

\subsection{Proof of Theorem \ref{thm:sellregretML}}
\label{pf:thm:sellregretML}
This proof consists of two parts, namely bounding seller's regret, and bounding buyer's regret as well the ``balance'' of buyer's budget and ROI constraints.

\paragraph{Part 1. Bounding seller's regret. } Here, we only need to show that the buyer's strategy is $\xi$-adaptive (see Definition in \ref{def:adapt}), and the rest of the proof follows from Theorem \ref{thm:episodicprice}.

For notation convenience, fix some exploration or exploitation episode $\mathcal{E}$, and denote the corresponding price in the episode as $d$.  In light of Lemma \ref{lem:thresholdOptSol}, we let $\bm{x}_{d}\in [0,1]^{N}$ be the unique optimal threshold vector (see Definition \ref{def:thresholdvect}) solution to $\util(d)$. According to the definition of the per-period seller expected revenue $\rev(d)$ under buyer best response
 in Equation (\ref{eq:revenue}), we can further write the seller's per-period expected revenue for episode $h$ as 
\begin{align}
\label{eq:benchmarkrev}
    \rev(d) =  d\sum_{n\in [N]} x_{d,n}g_{n}\,.
\end{align}

Let $\mathcal{F}_{t}$ be the sigma algebra generated by $\left\{(v_{\tau},d_{\tau},z_{\tau})\right\}_{\tau \in [t]}$, which characterizes all randomness in the buyer and seller's behavior up to period $t$. Recall $\Hat{\bm{x}}_{t}$ is the optimal solution to $U(d_{t})$ of Equation (\ref{eq:pricing:buyerutil}) via replacing the true distribution $\bm{g}\in \Delta_{N}$ with the estimate $\Hat{\bm{g}}_{t} \in \Delta_{N}$. The buyer adopting a threshold strategy w.r.t. $\Hat{\bm{x}}_{t}$ implies the buyer's decision to be
\begin{align}
\label{eq:buyerapprxaction}
    z_{t} = \sum_{n \in [N]}  \Hat{x}_{t,n}\I\{v_{t}= V_{n}\}
\end{align}
Since $\Hat{\bm{x}}_{t}$ is $\mathcal{F}_{t-1}$-measurable, for $t\in \mathcal{E}$ we have
\begin{align*}
 \E\left[z_{t}\Big |\mathcal{F}_{t-1}\right]
     ~=~ & \sum_{n\in [N]} g_{n}\Hat{x}_{t,n}
\end{align*}
Thus, the by defining
\begin{align}
    Y_{t}= \sum_{n \in [N]} g_{n} \Hat{x}_{t,n}  - z_{t}\,,
\end{align}
we know that the sequence $\{Y_{t}\}_{t\in \mathcal{E}}$ is a martingale difference sequence such that that $|Y_{t}|\leq 1$ for all $t$. By Azuma Hoeffding's inequality (see Lemma \ref{lem:azuma}) we have for any $\delta \in (0,1)$
\begin{align}
\label{eq:pricing:azuma}
    \prob\left(\widetilde{\mathcal{G}}\right) \geq 1-\delta \text{ where } \widetilde{\mathcal{G}} \defeq \left\{\left|\sum_{t\in \mathcal{E}}\left(\sum_{n \in [N]} g_{n} \Hat{x}_{t,n}  - z_{t}\right)\right| \leq \sqrt{2|\mathcal{E}|\log\left(2/\delta\right)}\right\} \,.
\end{align}

The remaining proof relies on the following lemma whose proof can be found in Appendix \ref{pf:lem:boundsol}
\begin{lemma}
\label{lem:boundsol}
Fix some price $d$ and define the following problem which is solved by the approximate best response buyer with ML advice to obtain $\Hat{\bm{x}}_{t}$ (see Definition \ref{def:approxbr}):
\begin{align}
\begin{aligned}
\label{eq:pricing:approxbuyerutil}
\Hat{\util}_{t}(d) =  \max_{\bm{x}\in [0,1]^{N}}  \sum_{n\in[N]} \Hat{\valpe}_{t,n}\V_{n} x_{n} 
 \quad \textrm{s.t.} &\sum_{n\in[N]} \Hat{\valpe}_{t,n}\left(\V_{n} - \gamma d\right)x_{n}\geq 0  \quad \text{ and }
 \quad d\sum_{n\in[N]} \Hat{\valpe}_{t,n} x_{n}\leq \rho \,.
 \end{aligned}
\end{align}
Here, recall $\Hat{\bm{g}}_{t} \in \Delta_{N}$ is the ML advice obtained in period $t$ which is an estimate for the true value distribution $\bm{g}\in \Delta_{N}$. Further, define the following 
values
\begin{align}
\label{eq:violvals}
\begin{aligned}
  (A) = \left(\util(d) -\sum_{n\in[N]} \valpe_{n}\V_{n} \Hat{x}_{t,n} \right)_{+},\quad (B) = \left(-\sum_{n\in[N]} \valpe_{n}\left(\V_{n} - \gamma d\right)\Hat{x}_{t,n}\right)_{+},\quad (C) = \left(d\sum_{n\in[N]} \valpe_{n} \Hat{x}_{t,n}- \rho\right)_{+}\,,
    \end{aligned}
\end{align}
where we recall $\util(d)$ is defined in Equation \eqref{eq:pricing:buyerutil}. Then, 
the values $(B), (C)$
 are upper bounded by $ \sqrt{N}\norm{\bm{g} - \Hat{\bm{g}}_{t}}$
for all $t$. Further, because the estimation error $\lim_{t\to \infty}\ell_{t} = 0$ there exists some $T_{0}\in \N$ s.t. $\norm{\bm{g} - \Hat{\bm{g}}_{t}}\leq \ell_{t} < \frac{g_{1}}{2}$ for all $t > T_{0}$. Then, there exists an absolute constant $C$ that only depends on buyer model primitives $(\bm{g},\bm{\V},\rho, \gamma)$ s.t. the values $(A)$ and $\norm{\bm{x}_{d}-\Hat{\bm{x}}_{t}}$ are  upper bounded by $ C\sqrt{N}\norm{\bm{g} - \Hat{\bm{g}}_{t}}$ for $t > T_{0}$, where $\bm{x}_{d}$ is the optimal solution to $\util(d)$.
\end{lemma}

We now show a high probability bound for  $\frac{\rev(d)}{d} -\sum_{t\in \mathcal{E}}z_{t}$. Assume event $\widetilde{\mathcal{G}}$ (Equation \eqref{eq:pricing:azuma}) holds, then
\begin{align*}
& \left|\sum_{t\in \mathcal{E}}\left(\frac{\rev(d)}{d} - z_{t}\right)\right|
    ~=~ \left|\sum_{t\in \mathcal{E}}\left(\sum_{n\in [N]} x_{d,n}g_{n}- z_{t}\right)\right| \\
   ~\leq~& \left|\sum_{t\in \mathcal{E}}\left(\sum_{n\in [N]} \Hat{x}_{t,n}g_{n} - z_{t}\right)\right| +
  \sum_{t\in \mathcal{E}} \left|\sum_{n\in [N]} \Hat{x}_{t,n}g_{n} - \sum_{n\in [N]} x_{d,n}g_{n} \right|\\
     ~\overset{(i)}{\leq}~ &
     \sqrt{2|\mathcal{E}|\log\left(2T\right)} + \sum_{t\in \mathcal{E}}  \norm{\bm{x}_{d} - \Hat{\bm{x}}_{t}} \cdot \norm{\bm{g}}\\
     ~\overset{(ii)}{\leq}~ &
     \sqrt{2|\mathcal{E}|\log\left(2T\right)}+ T_{0}  + C\sqrt{N}  \sum_{t\in \mathcal{E}:t > T_{0}} \ell_{t}\\
     ~\leq~ &
     \sqrt{2|\mathcal{E}|\log\left(2T\right)}+ T_{0}  + C\sqrt{N}  \sum_{t\in \mathcal{E}} \ell_{t}\\
     ~\overset{(iii)}{\leq}~ &
     \sqrt{2|\mathcal{E}|\log\left(2T\right)}+ T_{0}  + C\sqrt{N}  \widetilde{\phi}(|\mathcal{E}|)\\
     ~:=~  &\phi(|\mathcal{E}|)
\end{align*}
where in (i) we plugged in the Azuma-Hoeffding inequality result showed in Equation (\ref{eq:pricing:azuma}) with $\delta= \frac{1}{T}$; in (ii) we applied Lemma \ref{lem:boundsol} and some constant absolute constant $C$ for $t > T_{0}$ (defined in statement of Lemma \ref{lem:boundsol}), and the fact that $\norm{\bm{g}}\leq 1$ since $\bm{g}$ is a probability simplex; in (iii) we used the assumption that there exists some increasing function $\widetilde{\phi}$ s.t.
$ \sum_{t\in \mathcal{E}} \ell_{t} \leq \widetilde{\phi}(|\mathcal{E}|)$. Therefore w.p. at least $1-1/T$ (since $\widetilde{\mathcal{G}}$ holds w.p. at least $1-1/T$ when $\delta = 1/T$), we have
\begin{align*}
    \left| \frac{d}{|\mathcal{E}|} \sum_{t\in \mathcal{E}}z_{t} - \rev(d)\right| \leq \frac{\phi (|\mathcal{E}|)}{|\mathcal{E}|}
\end{align*}

Since $\widetilde{\phi}(x) \leq \mathcal{O}(x^{1-\tote})$, we know that $\phi(x) = \mathcal{O}(x^{1-\xi})$ for $\xi = \min\{\frac{1}{2},\tote\}$. Hence, for large enough $T$ s.t. the exploration episode length $E = T^{1-\xi + \epsilon} > T_{0}$, the buyer's approximate best responding with ML advice is $1-\xi$-adaptive for $\xi= \min\{\frac{1}{2},\tote\}$.

\paragraph{Part 2. Bounds for the buyer.}

We first follow a similar approach as the proof of Theorem \ref{corr:bestRespondBuyer} to upper bound the buyer regret. 

Let $d$ be the posted price in the final exploitation episode (see Algorithm \ref{alg:pricing}), and denote $\mathcal{E} = \Theta(T^{1-\xi + \epsilon})$ as all periods within exploration episodes. Then using the same arguments as in Equations \eqref{eq:brsingleutil} and \eqref{eq:brbenchmarkUB}, we can show the buyer's hindsight objective can be bounded as
\begin{align*}
 \opt(\bm{d}_{1:T}) 
    ~\leq ~&  \Theta(T^{\xi+\epsilon}) + (T-|\mathcal{E}|)\util(d)\,.
\end{align*} 
Since the buyer approximately best responds w.r.t. $\Hat{\bm{x}}_{t}$ which is the optimal solution to the problem $\Hat{\util}_{t}(d)$ Equation \eqref{eq:pricing:approxbuyerutil}, recall the buyer's decision $z_{t} $ can be written as in Equation \eqref{eq:buyerapprxaction}:
\begin{align*}
    z_{t} = \sum_{n \in [N]}  \Hat{x}_{t,n}\I\{v_{t}= V_{n}\}
\end{align*}
Hence, $\E[v_{t}z_{t} | \mathcal{F}_{t-1}] = \sum_{n\in [N]}\valpe_{n}\V_{n} \Hat{x}_{t,n}$. Let $C$ and $T_{0}$ be defined as in Lemma \ref{lem:boundsol},  and thus the buyer's regret can be thus bounded as followed
\begin{align*}
     \buyreg ~=~& \opt(\bm{d}_{1:T}) - \sum_{t\in [T]}\E\left[v_{t}z_{t}\right]\\
     ~\leq~& \Theta(T^{\xi+\epsilon})  + \sum_{t\in [T]/\mathcal{E}}\E\left[\left(\util(d) -g_{n}\Hat{x}_{t,n} \right)\right]\\
    ~\leq~& \Theta(T^{\xi+\epsilon}) + T_{0} + \sum_{t\in [T]/\mathcal{E}:t>T_{0}}\E\left[\left(\util(d) -g_{n}\Hat{x}_{t,n} \right)\right]\\
 ~\overset{(i)}{\leq}~& \Theta(T^{\xi+\epsilon}) + T_{0} + C\sqrt{N}\sum_{t\in [T]/\mathcal{E}:t>T_{0}}\norm{\bm{g} - \Hat{\bm{g}}_{t}}\\
 ~\overset{(ii)}{\leq}~& \Theta(T^{\xi+\epsilon}) + T_{0} + C\sqrt{N}\sum_{t\in [T]/\mathcal{E}}\ell_{t}\\
 ~\overset{(iii)}{\leq}~& \Theta(T^{\xi+\epsilon}) + T_{0} + C\sqrt{N}\widetilde{\phi}(T-|\mathcal{E}|)\\
 ~\overset{(iv)}{=}~& \Theta(T^{\xi+\epsilon})\,.
 \end{align*}
In (i), we applied Lemma \ref{lem:boundsol} for the value (A) defined in Equation \eqref{eq:violvals}; (ii) follows from the definition of the estimation errors $\ell_{t} \geq\norm{\bm{g} - \Hat{\bm{g}}_{t}} $; (iii) follows from the assumption that for any exploration or exploitation episode $\mathcal{E}_{h}$, the total error $\sum_{t\in \mathcal{E}_{h}}\ell_{t}$ is upper bounded by  $\widetilde{\phi}(T-|\mathcal{E}|$ where $\widetilde{\phi}$ is an increasing function; (iv) follows from the fact that 
$\widetilde{\phi}(x) \leq \mathcal{O}(T^{1-\tote})\leq \mathcal{O}(T^{1-\xi})$.

Now we show the buyer constraint violation 
is small, namely
\begin{align*}
    \frac{1}{T} \E\left[ \sum_{t\in [T]} \left(v_{t} - \gamma d_{t}\right)z_{t}\right]\geq - \Theta(T^{-\tote})\quad \text{and}\quad
     \frac{1}{T} \E\left[\sum_{t\in [T]}  d_{t}z_{t} \right]\leq \rho + \Theta(T^{-\tote})\,.
 \end{align*}
 The proofs for both inequalities are very similar, so here we just show $ \frac{1}{T} \E\left[ \sum_{t\in [T]} \left(v_{t} - \gamma d_{t}\right)z_{t}\right]\geq - \Theta(T^{-\tote})$. Similar to the above where we bounded buyer's regret, we have 
$\E[\left(v_{t} - \gamma d_{t}\right)z_{t} | \mathcal{F}_{t-1}] = \sum_{n\in [N]}\valpe_{n}\left(\V_{n} -\gamma d\right)\Hat{x}_{t,n}$, and thus for all exploration and exploitation episodes $\mathcal{E}_{1} \ldots \mathcal{E}_{H}$ (assuming there are $H$ episodes), we have
\begin{align*}
    -\E\left[ \sum_{t\in [T]} \left(v_{t} - \gamma d_{t}\right)z_{t}\right] ~=~& \sum_{t\in [T]}\E\left[  -\left(\sum_{n\in [N]}\valpe_{n}\left(\V_{n} -\gamma d\right)\Hat{x}_{t,n}\right)\right] \\
    ~\leq~& \sum_{t\in [T]}\E\left[  \left(-\sum_{n\in [N]}\valpe_{n}\left(\V_{n} -\gamma d\right)\Hat{x}_{t,n}\right)_{+}\right]\\
    ~\overset{(i)}{\leq}~& \sqrt{N}\sum_{t\in [T]}\ell_{t}\\
    ~=~& \sqrt{N}\sum_{h \in [H]}\sum_{t\in \mathcal{E}_{h}}\ell_{t}\\
     ~\overset{(ii)}{\leq}~& \sqrt{N}\sum_{h \in [H]}\mathcal{O}(|\mathcal{E}_{h}|^{1-\tote})\\
      ~=~& \Theta(T^{1-\tote})
\end{align*}
where (i) follows from the upper bound of (B) (Equation \eqref{eq:violvals}) in Lemma \ref{lem:boundsol}; (ii) follows from the
assumption that for any exploration and exploitation episode $\mathcal{E}_{h}$ the errors $\{\ell_{t}\}_{t}$ satisfy  $\sum_{t\in \mathcal{E}_{h}} \ell_{t} \leq \widetilde{\phi}(|\mathcal{E}_{h}|)$ for some increasing function $\widetilde{\phi}:\R_{+}\to \R^{+}$ and $\widetilde{\phi}(x) \leq \mathcal{O}(x^{1-\tote})$. 

Finally, dividing both sides by $T$ yields the desired bound  $ \frac{1}{T} \E\left[ \sum_{t\in [T]} \left(v_{t} - \gamma d_{t}\right)z_{t}\right]\geq - \Theta(T^{-\tote})$.
\qed

\subsection{Proof of Theorem \ref{thm:sellregretMLee}}
\label{pf:thm:sellregretMLee}
We know that the empirical estimates $\Hat{\bm{g}}_{t}\in \Delta_{N}$ for the buyer's value distribution $\bm{g}\in \Delta_{N}$ defined in Equation \eqref{eq:empiricalestimates} follow a multinomial distribution, i.e. $\Hat{\bm{g}}_{t} \sim \frac{1}{t}\text{Multinomial}(t,\bm{g})$. Therefore, applying Lemma \ref{lem:eeconc} by taking $\delta = 1/T^{2}$ , we have w.p. at least $1-1/T^{2}$ the following event  holds
\begin{align}
    \mathcal{G}_{t}:=\left\{\norm{\Hat{\bm{g}}_{t}-\bm{g}} \leq \ell_{t} :=  \sqrt{\frac{2N\log(2T^{2})}{t}}\right\}
\end{align}
Here we used the fact that $\norm{\bm{x}}\leq \norm{\bm{x}}_{1}$ for any vector $\bm{x}$. Hence, using a simple union bound, the event $ \cup_{t\in [T]}\mathcal{G}_{t}$ holds w.p. at least $1-1/T$. Further, for any exploration or exploitation episode $\mathcal{E}$, we have
\begin{align}
    \sum_{t\in \mathcal{E}}\ell_{t}  \leq 
     \sum_{\tau\in [|\mathcal{E}|]} \sqrt{\frac{2N\log(2T^{2})}{\tau}} \leq \widetilde{\phi}(|\mathcal{E}|) 
\end{align}
for some increasing function $\widetilde{\phi}$ s.t. $\widetilde{\phi}(x) \leq \mathcal{O}(x^{\frac{1}{2}})$. Hence, w.p. at least $1-1/T$, the estimation errors $\{\ell_{t}\}_{t}$ defined above satisfy the conditions in Theorem \ref{thm:sellregretML} for large enough $T$, i.e. $\lim_{t\to \infty}\ell_{t} = 0$ and $\sum_{t\in \mathcal{E}} \ell_{t} \leq \widetilde{\phi}(|\mathcal{E}|)$ for any any exploration or exploitation episode $\mathcal{E}$ where increasing function $\widetilde{\phi}:\R_{+}\to \R^{+}$ and $\widetilde{\phi}(x) \leq \mathcal{O}(x^{1-\tote})$. The rest of the proof directly follows from Theorem \ref{thm:sellregretML}.

\qed

\subsection{Proof of Lemma \ref{lem:boundsol}}
\label{pf:lem:boundsol}

Consider the region 
\begin{align}
    \mathcal{C} = \left\{\bm{x}\in [0,1]^{N}: 
-\sum_{n\in[N]} \valpe_{n}\V_{n} x_{n} \leq -\util(d), -\sum_{n\in[N]} \valpe_{n}\left(\V_{n} - \gamma d\right)x_{n} \leq 0, d\sum_{n\in[N]} \valpe_{n} x_{n}\leq \rho \right\}
\end{align}
By Lemma \ref{lem:thresholdOptSol}, we know that $\bm{x}_{d}$ is the unique optimal solution to $\util(d)$, and hence $\mathcal{C}$ consists of the single point 
$\bm{x}_{d}$, namely $\mathcal{C} = \{\bm{x}_{d}\}$. Now consider the optimal solution $\Hat{\bm{x}}_{t} \in [0,1]^{N}$ to $\Hat{\util}_{t}(d)$ in Equation \eqref{eq:pricing:approxbuyerutil}, by the Hoffman bound (Lemma \ref{lem:hoffman}), there exists some constant $H > 0$ that only depends on $(\bm{g},\bm{V})$ s.t.
\begin{align}
\label{eq:solapproxgood}    \norm{\Hat{\bm{x}}_{t}- \bm{x}_{d}} \leq H \left( \underbrace{\left(\util(d) -\sum_{n\in[N]} \valpe_{n}\V_{n} \Hat{x}_{t,n} \right)_{+}}_{(A)} + \underbrace{\left(-\sum_{n\in[N]} \valpe_{n}\left(\V_{n} - \gamma d\right)\Hat{x}_{t,n}\right)_{+}}_{(B)} + \underbrace{\left(d\sum_{n\in[N]} \valpe_{n} \Hat{x}_{t,n}- \rho\right)_{+}}_{C} \right)
\end{align}
where we used the inequality $\norm{(\bm{y})_{+}}\leq \sum_{n \in [N]}(y_{n})_{+}$ for any vector $y\in \R^{N}$. We now bound (A), (B) and (C) respectively.

\paragraph{Bounding (A).}
Similar to the proof of Theorem \ref{corr:bestRespondBuyer}, strong duality holds for the LP $\Hat{\util}_{t}(d)$, and hence there exists optimal dual variables $\Hat{\lambda},\Hat{\mu} \in \R_{+}$ s.t. 
\begin{align}
\label{eq:approxutilsd}
     \Hat{\util}_{t}(d) = \sum_{n\in[N]} \valpe_{n}\V_{n} \Hat{x}_{t,n}  = \max_{\bm{x}\in[0,1]^{N}}\sum_{n \in [N]}\Hat{g}_{t,n}\left( (1+\Hat{\lambda})V_{n} - (\gamma \Hat{\lambda} + \Hat{\mu}) d\right)x_{n} + \rho \Hat{\mu}
\end{align}
Since $\Hat{\util}_{t}(d) \leq 1$, it is easy to see $\Hat{\mu} \in [0,1/\rho]$, and further by considering $\bm{x} = (1, 0\ldots 0) \in \R^{N}$, we have 
\begin{align}
\label{eq:boundDV}
\begin{aligned}
& 1\geq \Hat{\util}_{t}(d)~ \geq~ \Hat{g}_{t,1}\left( (1+\Hat{\lambda})V_{1} - (\gamma \Hat{\lambda} + \Hat{\mu}) d\right)
~\overset{(i)}{\geq}~  \Hat{g}_{t,1} \Hat{\lambda}(V_{1} -\gamma d)  - \Hat{\mu}d 
~\overset{(ii)}{\geq}~  \frac{g_{1}}{2} \cdot  \Hat{\lambda}(V_{1} -\gamma d)  - \frac{1}{\rho} \\
~\overset{(iii)}{\Rightarrow} ~ &  \Hat{\lambda}\leq 
2\left(1 +\frac{1}{\rho}\right)\frac{\V_{1}-\gamma D_{1}}{g_{1}}\,,
\end{aligned}
\end{align}
where in (i) we used the fact that $\Hat{g}_{t,1}\in [0,1]$; in (ii) we used the fact that $ |\Hat{g}_{t,1} - g_{1}|\leq \norm{\Hat{\bm{g}}_{t} - \bm{g}} \leq \ell_{t} < \frac{g_{1}}{2}$ for all $t > T_{0}$, and also $d\in [0,1]$ as well as $\Hat{\mu} \in [0,1/\rho]$; in (iii), we used Assumption \ref{assum:price:nontriv} s.t. $V_{1} - \gamma d > 0$ for all $d \in \mathcal{D}$, and $g_{1} > 0$.

On the other hand, we have
\begin{align}
\begin{aligned}
\label{eq:approxdifobj}
    \util(d) ~\leq~ &  \max_{\bm{x}\in[0,1]^{N}}\sum_{n \in [N]}g_{n}\left( (1+\Hat{\lambda})V_{n} - (\gamma \Hat{\lambda} + \Hat{\mu}) d\right)x_{n} + \rho \Hat{\mu}\\
    ~\overset{(i)}{\leq}~ &  \max_{\bm{x}\in[0,1]^{N}}\sum_{n \in [N]}\Hat{g}_{t,n}\left( (1+\Hat{\lambda})V_{n} - (\gamma \Hat{\lambda} + \Hat{\mu}) d\right)x_{n}  + (1+\Hat{\lambda}) \sum_{n\in [N]} |\Hat{g}_{t,n} - g_{n}|+\rho \Hat{\mu}\\
    ~\overset{(ii)}{\leq}~ &  \max_{\bm{x}\in[0,1]^{N}}\sum_{n \in [N]}\Hat{g}_{t,n}\left( (1+\Hat{\lambda})V_{n} - (\gamma \Hat{\lambda} + \Hat{\mu}) d\right)x_{n} +\rho \Hat{\mu} + (1+\Hat{\lambda})\sqrt{N}\norm{\Hat{\bm{g}}_{t}-\bm{g}} \\
     ~\overset{(iii)}{=}~ &   \Hat{\util}_{t}(d)  + 2\left(1 +\frac{1}{\rho}\right)\frac{\V_{1}-\gamma D_{1}}{g_{1}}\cdot \sqrt{N}\norm{\Hat{\bm{g}}_{t}-\bm{g}}
    \end{aligned}
\end{align}

In (i), we used the fact that for all $n \in [N]$, $x_{n}\in [0,1] $ and $(1+\Hat{\lambda})V_{n} - (\gamma \Hat{\lambda} + \Hat{\mu}) d \leq (1+\Hat{\lambda})V_{n} \leq 1+\Hat{\lambda}$ since all possible values $V_{n} \in [0,1]$; (ii) applies Cauchy–Schwarz inequality; (iii) plugs in Equation \eqref{eq:approxutilsd} and \eqref{eq:boundDV}.

Therefore, if $\Hat{\util}_{t}(d) = \sum_{n\in[N]} \valpe_{n}\V_{n} \Hat{x}_{t,n} \geq \util(d)$, then (A) = 0, whereas if $\Hat{\util}_{t}(d) = \sum_{n\in[N]} \valpe_{n}\V_{n} \Hat{x}_{t,n} < \util(d)$, Equation \eqref{eq:approxdifobj} implies 
\begin{align}
\label{eq:approxdifobj1}
    (A)\leq 2\left(1 +\frac{1}{\rho}\right)\frac{\V_{1}-\gamma D_{1}}{g_{1}}\sqrt{N}\norm{\Hat{\bm{g}}_{t}-\bm{g}}
\end{align}

\paragraph{Bounding (B) and (C).} The bounds for (B) and (C) are similar, and therefore we only show that for (B).
\begin{align}
\label{eq:ROIapproxgood}
\begin{aligned}
   (B)= \left( -\sum_{n\in[N]} \valpe_{n}\left(\V_{n} - \gamma d\right)\Hat{x}_{t,n}\right)_{+}
    ~\overset{(i)}{\leq}~& \left(-\sum_{n\in[N]} \Hat{\valpe}_{t,n}\left(\V_{n} - \gamma d\right)\Hat{x}_{t,n} \right)_{+} + \left|\sum_{n\in[N]}\left(\Hat{\valpe}_{t,n} - \valpe_{n} \right)\left(\V_{n} - \gamma d\right)\Hat{x}_{t,n}\right|\\
     ~\overset{(ii)}{\leq}~& \sum_{n\in[N]}|\Hat{\valpe}_{t,n} - \valpe_{n}| \\
      ~\overset{(iii)}{\leq}~& \sqrt{N}\norm{\Hat{\bm{g}}_{t}-\bm{g}}
     \end{aligned}
\end{align}
Here, (i) follows from the basic inequality sequence $(a+b)_{+}\leq (a)_{+} + (b)_{+} \leq (a)_{+} + |b|$; (ii) follows from the fact that $\Hat{\bm{x}}_{t}$ is feasible to $\Hat{\util}_{t}(d)$ so that $\sum_{n\in[N]} \Hat{\valpe}_{t,n}\left(\V_{n} - \gamma d\right)\Hat{x}_{t,n} \geq 0$, and also
$|V_{n}- \gamma d|\leq V_{n}\leq 1$ and $\Hat{x}_{t,n} \in [0,1]$; (iii) follows from the Cauchy–Schwarz inequality.

We can similarly show 
\begin{align}
\label{eq:budgetapproxgood}
\begin{aligned}
   (C)\leq \sqrt{N}\norm{\Hat{\bm{g}}_{t}-\bm{g}}
     \end{aligned}
\end{align}

Finally, combining Equations \eqref{eq:solapproxgood}, \eqref{eq:approxdifobj1}, \eqref{eq:ROIapproxgood}, and \eqref{eq:budgetapproxgood} yields the desired result.
\qed

\section{Supplementary Lemmas}



\begin{lemma}[Azuma–Hoeffding inequality]
\label{lem:azuma}
beLet $Y_{1} \dots Y_{n}$  be a martingale difference sequence with a uniform bound $|Y_{j}| \leq 1$ for all $j\in [n]$. Then  for any $\delta \in (0,1/e)$,
\begin{align*}
    \prob\left(\left|\sum_{j\in[n]}Y_{j}\right|  > \sqrt{2n\log(2/\delta)} \right) \leq \delta\,.
\end{align*}

\end{lemma}

\begin{lemma}[Hoffman bound \cite{hoffman2003approximate}]
\label{lem:hoffman}
Consider a non-empty linear region $\mathcal{C} = \left\{\bm{x}\in \R^{n}: \bm{A}\bm{x} \leq \bm{b}\right\}$ for some $\bm{b} \in \R^{n}$ and $\bm{A} \in \R^{m\times n}$. Then, there exists some constant $H >0 $ that only depends on $\bm{A}$ s.t. for any $\bm{y} \in \R^{n}$ we have $\inf_{\bm{z} \in \mathcal{C}} \norm{\bm{z} - \bm{y}}\leq H \norm{(\bm{A}\bm{y} - \bm{b})_{+}}$. Here $(\bm{y})_{+}$ is the vector that takes the positive parts for each entry in $\bm{y}$, i.e. $(\bm{y})_{+} = \left((y_{1})_{+}\ldots (y_{n})_{+}\right)$.
\end{lemma} 

\begin{lemma}[Empirical distribution concentration inequality \cite{weissman2003inequalities}]
\label{lem:eeconc}
Let $\bm{g}\in \Delta_{N}$ be a $N$-dimensional probability simplex $(N\geq 2)$, and $\Hat{\bm{g}}_{t} \sim \frac{1}{t}\text{Multinomial}(t,\bm{g})$. Then for any $\delta \in (0,1)$, we have
\begin{align*}
    \prob\left(\norm{\Hat{\bm{g}}_{t}-\bm{g}}_{1}> \sqrt{\frac{2N\log(2/\delta)}{t}}\right) \leq \delta
\end{align*}
\end{lemma}
See also \cite{qian2020concentration} Proposition 2. for a similar statement to that of Lemma \ref{lem:eeconc}.

\end{APPENDICES}

\end{document}